\documentclass[11pt]{article}

\usepackage{algorithm} 
\usepackage{algpseudocode}
\usepackage{amsthm}
\usepackage{authblk}
\usepackage{float}
\usepackage[T1]{fontenc}
\usepackage[margin=1in]{geometry}
\usepackage{graphicx}
\usepackage{hyperref}
\usepackage[latin1]{inputenc}
\usepackage{mathtools,amsthm,amssymb}
\usepackage{mathabx}
\usepackage{mathrsfs}
\usepackage{natbib}
\usepackage{subfig}

\numberwithin{equation}{section}

\newcommand{\dif}{\mathrm{d}}
\newtheorem{theorem}{Theorem}

\newtheorem{proposition}{Proposition}[section]
\newtheorem{lemma}{Lemma}[section]
\newtheorem{remark}{Remark}[section]
\newtheorem{definition}{Definition}[section]
\newtheorem{example}{Example}[section]

\usepackage{tabularx,ragged2e,booktabs}

\newcolumntype{L}{>{\RaggedRight\arraybackslash}X}
\newcolumntype{C}[1]{>{\Centering\arraybackslash%
        \hsize=#1\hsize\linewidth=\hsize}X}

\newcommand{\tildeb}[1]{#1}

\newcommand{\Circle}{\mathrm{circle}}
\newcommand{\Bounce}{\mathrm{bounce}}

\providecommand{\keywords}[1]
{
  \small	
  \textbf{\textit{Keywords---}} #1
}

\title{Haar-Weave-Metropolis kernel}

\author[1]{Kengo Kamatani\thanks{kamatani@ism.ac.jp}}
\author[2]{Xiaolin Song\thanks{Corresponding Author: songxl@sigmath.es.osaka-u.ac.jp}}
\affil[1]{Institute of Statistical Mathematics}
\affil[2]{Graduate School of Engineering Science, Osaka University}
\date{}
\begin{document}

\maketitle

\begin{abstract}
Recently, many Markov chain Monte Carlo methods have been developed with deterministic reversible transform proposals inspired by the Hamiltonian Monte Carlo method. The deterministic transform is relatively easy to reconcile with the local information (gradient etc.) of the target distribution.  However, as the ergodic theory suggests, these deterministic proposal methods seem to be incompatible with robustness and lead to poor convergence, especially in the case of target distributions with heavy tails.  On the other hand, the Markov kernel using the Haar measure is relatively robust since it learns global information about the target distribution introducing global parameters. However, it requires a density preserving condition, and many deterministic proposals break this condition. In this paper, we carefully select deterministic transforms that preserve the structure and create a Markov kernel, the Weave-Metropolis kernel, using the deterministic transforms. By combining with the Haar measure, we also introduce the Haar-Weave-Metropolis kernel. In this way, the Markov kernel can employ the local information of the target distribution using the deterministic proposal, and thanks to the Haar measure, it can employ the global information of the target distribution.  Finally, we show through numerical experiments that the performance of the proposed method is superior to other methods in terms of effective sample size and mean square jump distance per second.
\end{abstract}

%TC:ignore
\keywords{Markov chain, Hamiltonian Monte Carlo, Haar measure, Bayesian inference}

\section{Introduction}

The fundamental object of Bayesian statistics is the posterior distribution, and all statistical inference is performed by integrating a given quantity with respect to the posterior distribution. Therefore, the evaluation of the integral is the main obstacle in Bayesian statistics, and Bayesian computational methods have been developed for this purpose.  The Markov chain Monte Carlo method, especially, the Metropolis--Hastings kernel has been the gold standard for Bayesian computation for the last thirty years. However, its efficiency seems to be diminishing due to the arise of modern big data with complex dependent models. 
Most Bayesian computational methods work efficiently for high-dimensional models in the 1990 sense, but have difficulty in the modern era. Bayesian computation needs a breakthrough to keep pace with change.

Recently, the Metropolis--Hastings kernel based on deterministic reversible transform have been developed such as \cite{murray2010, nice_mcmc, levy2017, zhang2018theory}. %These new methods are studied in the field of Bayesian computation as well as machine learning. 
%In general, the flow refers to a semigroup of mappings from the state space to itself, possibly with some kinds of smoothness. 
Typical, and the earliest example is the  Hamilton Monte Carlo method \citep{DUANE1987216}. Reversible transform-based methods can use local information, that is, the gradients of the target probability density function. In addition, these methods can train the transition kernels, e.g., through the neural network strategy using the generated random sequences. However, due to the lack of global information, reversible transform-based methods can lead to poor performance because the local information is sometimes less informative, e.g., for  heavy-tailed target distributions. 
% In \cite{zhang2018theory}, they also built the concept to construct the Markov kernel by use ergodic flow, the hyper parameters of the flow based Markov kernel are tuned by use variational inference. 

On the other hand, it is possible to use global information of the target distribution by introducing a global parameter into the Metropolis--Hastings kernel. The global parameters are estimated from random samples. If there is a sufficient information, we can plug in this information to the estimator.  However, this is not a robust strategy because a poor estimate leads directly to poor convergence.  It is advisable to set a non-informative prior as it leads to robust choice. In this purpose, it is natural to use the Haar measure for the prior distribution of the parameter. The usefulness of Haar measure has been analysed such as \cite{Liu_1999, 10.1093/biomet/87.2.353, MR3231603, Shariff2015, MR3668488, Kamatani_2018}. 
This strategy sometimes improves the performance of Markov chain Monte Carlo drastically. The improvement is theoretically proved for a specific kernel in terms of ergodic property \citep{MR3668488} and high-dimensional convergence speed \citep{Kamatani_2018}. In this paper, we will provide a general theory (Theorem \ref{thm:uniform_ergodicity}) that partly explain the benefit of the use of Haar measure.

%Flow
Therefore, on the one hand, there is an efficient method based on the reversible transform, which takes into account the local information of the target distribution. On the other hand, there is a method based on the Haar measure that uses global information. It is natural to consider a combination of local and global information methods to solve complicated problems. However, this is not an easy task because, at least in our framework, the Haar measure-based methods require that the transform is measure preserving respect to a probability measure. In other words, it should satisfy a density preserving condition (see the paragraph after Definition \ref{def:ahm}). The locally-informed reversible transform usually destroys this structure. In this paper, we carefully select reversible transforms that satisfy the density preserving condition. To be more precise, we uses circle transform and bounce transform as reversible transform-based updates.

As in the Hamiltonian Monte Carlo kernel, the reversible transform is defined for the variable $(x,v)$, where $x$ is the state variable and $v$ is an auxiliary variable. If we focus on the behaviour of $x$, then the path of $x$ evolves with the elliptic motion induced by the circular transform. After a certain time, the bounce forces $x$ to follow another ellipse. Because of this property, the path looks like a weaving behaviour. For this reason, we call this transform the Weaving transform. See Figure \ref{fig::path} for a typical behaviour of the path of the Weave transform.  

Thanks to the bounce transform going in the opposite direction to the gradient, the Weaving transform does not change the corresponding potential energy (negative log-likelihood) as much.  This is similar to the Hamiltonian flow, which does not change the value of the Hamiltonian as much. In addition, the Weave transform does not change the distance from the origin. This avoids an unfavourable behaviour for the super-light target distribution \citep [see, e.g., Theorem 4.2 of][]{RT2}. 
The Weave transform itself is not sufficient to move throughout the state space. 
This local behaviour, reinforced by the Haar motion, traverses the state space quite well. 

\begin{figure}%
    \centering
    {{\includegraphics[width=0.4\textwidth]{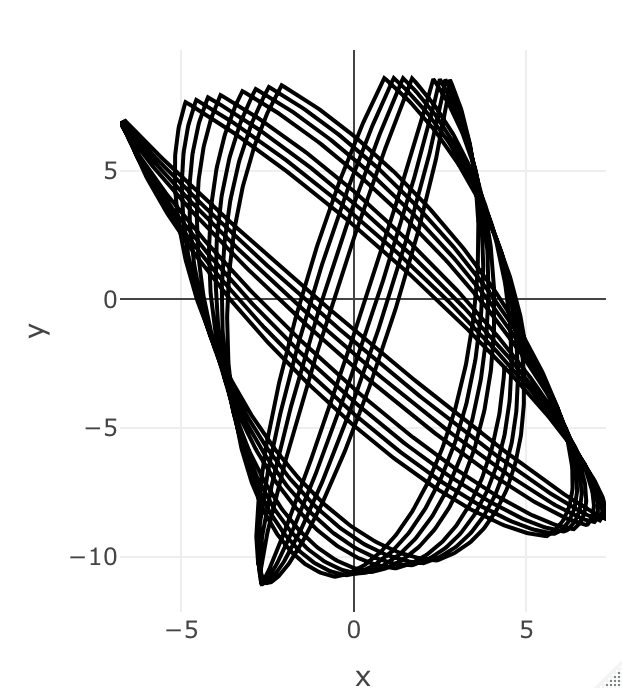} }}%
    \qquad
    {{\includegraphics[width=0.47\textwidth]{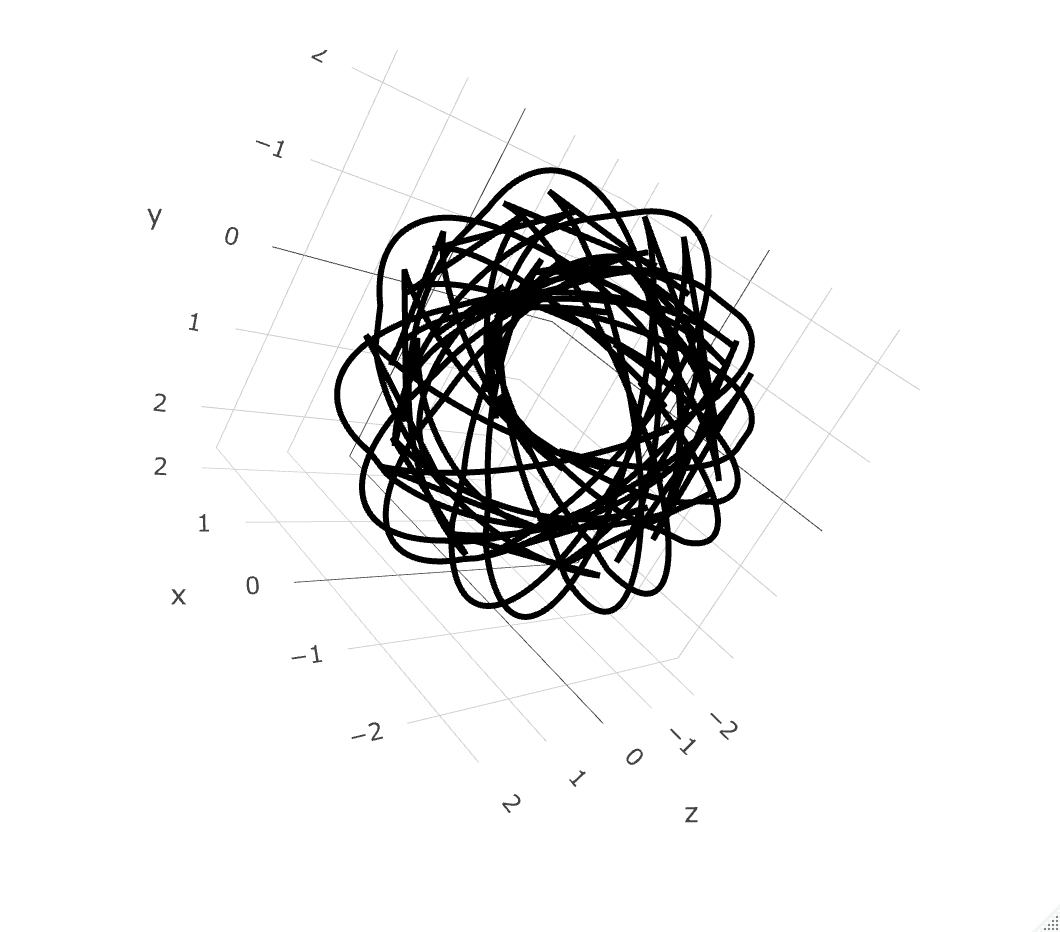} }}%
    \caption{Both the left and right panels show the trace of the Weave transform  with step number $L=40$. The left panel is target on a $2-$dimensional Student $t$-distribution and the right panel is toward on $3-$dimensional Student $t$-distribution.}%
    \label{fig::path}
\end{figure}

% In \cite{murray2010} and \cite{MR3231603}, the elliptical slice sampler also search the propose state by move along the ellipse. The Weave type kernel can be more efficient since it change the direction of ellipse frequently by incorporate with the gradient information(bounce movement). We confirmed the Haar mixture version of Weave metropolis algorithm's convergence property....... Comparing with other gradient based MCMC algorithms, we find that the Haar mixture version of Weave metropolis algorithm behaves more robust for heavy tail distributions in different kinds of numerical experiments.
% \kengo{Comment on Local (bounce) and Global (Haar) motions. The circle transform is essential to connect these two completely different kernels...}

% \begin{figure}[h]
% \centering
% \includegraphics[width=0.8\textwidth]{Rplot.png}
% \caption{Both the left and right panels show the trace of propose method in single iteration towards on a $2$-dimensional Student distribution.  They start from the same start point(blue point) with different initial velocities.}
% \label{fig::path}
% \end{figure}

The paper is organised as follows. Section 2 gives a brief overview of reversible Markov kernel and
augmented Metropolis--Hastings kernel. In Section 3, we introduce the Haar measure and the Haar mixture Metropolis kernel. In addition, we also explain the advantage of using heavy-tailed reference measure. In Section 4, we propose the Weave-Metropolis and Haar-Weave-Metropolis algorithms that use the Weave transform. Section 5 studies the limit behaviour of the weave processes. In Section 6, we compare the proposed algorithms with other Markov chain Monte Carlo methods through some numerical experiments.

\subsection{Notation}
Write $\mathcal{N}_d(\mu,\Sigma)$ for the $d$-dimensional normal distribution with mean $\mu$ and variance $\Sigma$. Write $\mathcal{N}(\mu,\Sigma)$ when $d=1$. Similarly, we write $\mathcal{G}(\nu,\alpha)$ for the gamma distribution with shape parameter $\nu$ and rate parameter $\alpha$, and we write $\mathcal{B}e(\alpha_1,\alpha_2)$ 
for the Beta distribution with shape parameters $\alpha_1$ and $\alpha_2$. Let $\mathcal{T}_{d}(\nu, \mu,\Sigma)$ be a $d$-dimensional multivariate Student $t$-distribution with mean (shift parameter) $\mu$, variance (scale parameter) $\Sigma$ and degree of freedom $\nu$.

For a vector $x=(x_1,\ldots, x_d)\in\mathbb{R}^d$, $|x|=(\sum_{i=1}^d x_i^2)^{1/2}$ is the Euclidean norm. If $A$ is a $d\times d$-square matrix, $|A|=\sup_{|x|=1, x\in\mathbb{R}^d}|Ax|$. 
$I_d$ is the $d\times d$-identity matrix. 
% $$
% |Ax|^2=(Ax)^{\top}Ax=x^{\top}A^{\top}Ax\le \lambda(A^{\top}A)
% $$

\section{Reversible transform based Metropolis--Hastings kernel}

\label{subsec:hmc_hug}

A Markov kernel $Q$ on $(E,\mathcal{E})$ is called $\mu$-reversible for a measure $\mu$ if
$$
\int_A\mu(\dif x)Q(x,B)=\int_B\mu(\dif x)Q(x,A)
$$
for any $A, B\in\mathcal{E}$. 
The Metropolis kernel, defined as follows, is $\Pi$-reversible  where $\Pi(\dif x)=\exp(-U(x))\mu(\dif x)$ is a probability measure. 
The definition here is somewhat broader than usual, which is useful for our purpose. We want to define the Metropolis kernel as a subclass of the Metropolis--Hastings kernel such that the acceptance rate can be written as a fraction of the target densities. 

\begin{definition}[Metropolis kernel]\label{def:metropolis}
Let $\mu$ be a $\sigma$-finte measure. 
For a $\mu$-reversible Markov kernel $Q$ on $(E,\mathcal{E})$, a Markov kernel $P$ defined by
\begin{align*}
    P(x,A)=\delta_x(A)\left(1-\int_{y\in E} Q(x,\dif y)\alpha(x,y)\right)+\int_A Q(x,\dif y)\alpha(x,y)
\end{align*}
is called the Metropolis kernel, where $\delta_x$ is the Dirac measure on $x\in E$, and 
\begin{equation}\label{eq:acceptance_probability}
    \alpha(x,y)=\min\left\{1,\exp(-U(y)+U(x))\right\}. 
\end{equation}
The Markov kernel $Q$ is called the proposal kernel, the function $\alpha$ is called the acceptance probability, and the measure $\mu$ is called the reference measure for the Metropolis kernel. 
\end{definition}

We call any one-to-one measurable mapping $\Phi$ from $E$ to itself a transform. Furthermore, we call a transform gradient-based if it uses gradient information of the target density. For $L\in\mathbb{N}$, we write
$$
\Phi^L(x)=\overbrace{\Phi\circ\cdots\circ\Phi}^{L}(x). 
$$
We introduce a Metropolis kernel based on a transform. 
We call $\Phi$ $\mu$-measure preserving if
$$
\mu(\{x\in E:\Phi(x)\in A\})=\mu(A)
$$
for $A\in\mathcal{E}$. We also call $\Phi$
$\mu$-reversible transform if 
\begin{equation}\label{eq:reversible}
\mu(\{x\in A:\Phi(x)\in B\})=
\mu(\{x\in B:\Phi(x)\in A\})
\end{equation}
for $A, B\in\mathcal{E}$. 
Note that every $\mu$-reversible transform $\Phi$ is $\mu$-measure preserving, and if the $\mu$-measure preserving transform $\Phi$ satisfies the condition $\Phi\circ\Phi(x)=x$, then it is $\mu$-reversible. For the construction of a Metropolis kernel based on $\mu$-measure preserving transform, we need a slightly more general version of reversibility, flipped reversibility, which is sometimes called skew reversibility.  

\begin{definition}[Flipped reversible  transform]\label{def:semi-reversible-flow}
Let $\mu$ be a probability measure on $(E,\mathcal{E})$. 
Let $\kappa:E\rightarrow E$ be a $\mu$-measure preserving transform such that $\kappa\circ\kappa(x)=x\ (x\in E)$. 
We call $\Phi:E\rightarrow E$  $(\mu,\kappa)$-reversible transform if $\kappa\circ\Phi$ is $\mu$-reversible.   
\end{definition}

The properties of the flipped reversible transform are explored in Section \ref{sec:reversibility}. 
We will design a Metropolis--Hastings kernel, that uses a flipped reversible transform. 
Let us consider an extended state space $(E^2,\mathcal{E}^{\otimes 2})$ so that the construction of the flipped reversible transform becomes simpler in practice. Let 
$$
\mu^{\otimes 2}(\dif x,\dif v)=\mu(\dif x)\mu(\dif v),\ 
\mu\otimes\lambda(\dif x,\dif v)=\mu(\dif x)\lambda(\dif v). 
$$

\begin{definition}[Augmented Metropolis--Hastings kernel]
\label{def:mh_augmented}
Let $\kappa(x,v)=(x,-v)$. 
Let $\Phi:E^2\rightarrow E^2$ be $(\mu^{\otimes 2 },\kappa)$-reversible transform. 
Let $\Pi(\dif x)=\exp(-U(x))\mu(\dif x)$ and $\lambda(\dif v)=\exp(-K(v))\mu(\dif v)$ be  probability measures on $(E,\mathcal{E})$ such that $K(v)=K(-v)$. The augmented Metropolis--Hastings kernel is a Markov kernel on $(E,\mathcal{E})$ defined by 
\begin{equation}
\label{eq:metropolis_hastings_augmented_var}
\begin{split}
    P(x,A)&=\delta_x(A)\left(1-\int_{v\in E}\lambda(\dif v)
    \alpha((x,v),\Phi(x,v))\right)\\
    &\quad+\int_{v\in E} 1(\{\Phi(x,v)\in A\times E\})~\lambda(\dif v)
    \alpha((x,v),\Phi(x,v))
    \end{split}
\end{equation}
with the acceptance probability
$$
\alpha((x,v),(x^*,v^*))=\min\left\{1,\exp(-H(x^*, v^*)+H(x, v))\right\} 
$$
where the function $H(x,v)$, the so-called Hamiltonian, is defined by 
$$
H(x,v)=U(x)+K(v). 
$$
\end{definition}

See Section \ref{sec:reversibility} for the proof of $\Pi$-reversibility of $P$.  
The Hamiltonian Monte Carlo kernel belongs to this class. In this case $\mu$ is the Lebesgue measure on $(\mathbb{R}^d,\mathcal{B}(\mathbb{R}^d))$ and $\lambda$ is a normal distribution.  For $h>0$, the Hamiltonian transform $\Phi(x,v)=(x_h,v_h)$ is defined by the solution $(x_t,v_t)_{t\ge 0}$ of
$$
\frac{\dif x}{\dif t}=\frac{\partial H}{\partial v}, \quad
\frac{\dif v}{\dif t}=-\frac{\partial H}{\partial x}. 
$$
The Hamiltonian flow does not change the value of the Hamiltonian. Thanks to this property together with the $(\mathrm{Leb},\kappa)$-reversible property, the Hamiltonian transform is 
$(\Pi\otimes\lambda,\kappa)$-reversible and the acceptance probability is always $1$ (see Proposition \ref{prop:reversible_flow}). 
In practice, however, we need a discretisation, and the leap-frog approximation is commonly used.

\begin{example}[Leap-frog approximated Hamiltonian Monte Carlo kernel]\label{ex:leap_frog}
Let $h>0$. 
Suppose $U:\mathbb{R}^d\rightarrow\mathbb{R}$ and $K:\mathbb{R}^d\rightarrow\mathbb{R}$ are differentiable. 
The leap-frog approximated transform $\phi_h(x,v)=(x_h,v_h)$ is defined by
\begin{equation*}
\begin{cases}
 v_{h/2}&=v-\frac{h}{2}\nabla U(x) \\
    x_h&=x+h~\nabla K(v_{h/2})\\
    v_h&=v_{h/2}-\frac{h}{2}\nabla U(x_{h}). 
\end{cases}\label{eq:label2}
\end{equation*}
In practice, the iterated transform $\Phi(x,v)=(\phi_h)^L(x,v)$ is used, where $L\in\mathbb{N}$. 
The approximated transform no longer preserves the Hamiltonian and is therefore not $(\Pi\otimes\lambda,\kappa)$-reversible. However, it is still $(\mathrm{Leb},\kappa)$-reversible. Therefore, we can construct an augmented Metropolis--Hastings kernel where the underlying measure $\mu$ is the Lebesgue measure on $(\mathbb{R}^d,\mathcal{B}(\mathbb{R}^d))$. 
\end{example}
\begin{example}[Infinite dimensional Hamiltonian Monte Carlo kernel]
\label{ex:idhmc}
A version of the leap-frog transform, which replaces the middle step with the circular transform (Definition \ref{def:circular}) and replaces $U(x)$ by $U(x)+|x|^2/2$, is also $(\mathrm{Leb},\kappa)$-reversible. The corresponding Metropolis--Hastings kernel is called the infinite dimensional Hamiltonian Monte Carlo  kernel \citep{MR2858447, MR2822774}. 
Because of the similarity to the kernel we will introduce, we will focus on this kernel later in Section \ref{sec:simulation}. 
\end{example}

Now let us look at another closely related kernel, The Hug kernel introduced in \cite{ludkin2019hug}. The core for the kernel is the bounce transform.
Let  $\xi:\mathbb{R}^d\rightarrow \mathbb{R}^d$ be a vector-valued function. Let $\mathcal{N}=\{x\in\mathbb{R}^d:\xi(x)=0\}$. For any $x\in\mathcal{N}^c$, let
\begin{equation}\label{eq:bounce}
\overline{\xi}(x)=\frac{\xi(x)}{|\xi(x)|}. 
\end{equation}
The function $\xi(x)$ is indeterminate if  $x\in\mathcal{N}$. 

\begin{definition}(\textbf{bounce transform})
\label{def:bounce}
For any $x,v\in\mathbb{R}^d$, consider an operation
\begin{equation}
R(x)~v= v-2(\overline{\xi}(x)^\top v)\overline{\xi}(x)=(I-2\overline{\xi}(x)\overline{\xi}(x)^{\top})v
\label{bounce}
\end{equation}
if $x\notin\mathcal{N}$ and set $R(x)v=-v$ if $x\in\mathcal{N}$. 
We call $\Phi_{\Bounce}(x,v)=(x,R(x)v)$ the bounce transform. 
More generally, for a $d$-dimensional vector $M$ with a symmetric positive definite matrix $\Sigma$, we consider an operation 
$$
R(x\mid M, \Sigma)~v=M+\left(I-2\frac{\Sigma\xi(x)\xi(x)^{\top}}{\xi(x)^{\top}\Sigma\xi(x)}\right)(v-M)
$$
and set $\Phi_{\Bounce}(x,v\mid M, \Sigma)=(x, R(x\mid M, \Sigma)v)$. 
We have $\xi(x)^{\top}(R(x\mid M, \Sigma)v-M)=-\xi(x)^{\top}(v-M)$. In particular, $\xi(x)^{\top}R(x)v=-\xi(x)^{\top}v$. 
\end{definition}

The bounce transform reflects the velocity $v$ with respect to the hyper-plane orthogonal to $\xi(x)$. The bounce transform has been used in the Monte Carlo literature, for example in \cite{MR1994729, Peters:2012aa, 10.1093/biomet/asab013, ludkin2019hug}. Note that the general bounce operation $R(x\mid M,\Sigma)$ and the corresponding transform $\Phi_{\Bounce}(x,v\mid M,\Sigma)$ are used only in the simulation section. In the rest of the paper, for simplicity, we consider only $R(x)$ and $\Phi_{\Bounce}$.

 \begin{example}[Hug kernel]\label{ex:hug}
Let $h>0$. Assume that $U$ and $K$ are differentiable. 
The Hug transform $\phi_h(x,v)=(x_h,v_h)$ is defined by
\begin{equation*}
\begin{cases}
 x_{h/2}&=x+\frac{h}{2}\nabla K(v) \\
    v_h&=R(x_{h/2})v_{h/2}\\
    x_{h}&=x_{h/2}+\frac{h}{2}\nabla K(v_h). 
\end{cases}
\end{equation*}
This transform is $(\mathrm{Leb},\kappa)$-reversible. The augmented Metropolis--Hastings kernel using $\Phi=(\phi_h)^L$ will be referred to as the Hug kernel for $L\in\mathbb{N}$.  
\end{example}

Under suitable conditions, the transform $\Phi$ changes the potential energy $U(x)$ on the order of $h^2$ when $h\rightarrow 0$ \citep{ludkin2019hug}. On the other hand, the transform lacks the ability to efficiently traverse the state space, so it is slow to converge to equilibrium. In \cite{ludkin2019hug}, a different Markov kernel, called the Hop kernel, was introduced to support traveling in the state space.

\section{Haar mixture Metropolis kernel}

\subsection{Haar masure and Haar mixture kernel}
\label{sec:haar}

In this section, we recall the Haar mixture Metropolis kernel discussed in \cite{kamatani2020} that is a trick improving convergence of Metropolis kernels. 
Let $(G,\times)$ be a locally compact topological group with a Borel $\sigma$-algebra. We also assume that the topological group is unimodular, that is, the left Haar measure and the right Haar measure conicdes up to a multiplicative constant.  Let $\nu$ be the  Haar measure. The  Haar measure satifies $\nu(H)=\nu(gH)$ for every $g\in G$ and every Borel set $H$ of $G$. 
Let $E$ be an abelian topological group with a Borel $\sigma$-algebra $\mathcal{E}$. 
We assume that $E$ is a $G$-module, i.e., there is an operation $(g,x)\mapsto gx$ such that $ga+gb=g(a+b)$. A typical example of the pair $G$ and $E$ is $\mathbb{R}$ and $\mathbb{R}^d$ with a scalar multiplication $(a,x)\mapsto ax$ as the left group action. 

Let $\mu$ be a probability measure and 
let $Q(x,\dif y)$ be a $\mu$-reversible Markov kernel on $(E,\mathcal{E})$. 
Let $\mu_g(A)=\mu(gA)$ and 
let $Q_g(x,A)=Q(gx,gA)$. 
Let 
$$
\mu_*(A)=\int\mu_g(A)\nu(\dif g). 
$$
Assume that $\mu_*$ is $\sigma$-finite. Observe that $\mu_*$ is a left invariant measure. Let 
$$
K(x,\dif g)=\frac{\dif \mu_g}{\dif \mu_*}(x)\nu(\dif g). 
$$
By Fubini's lemma, $K(x,G)=1$ for $\mu_*$-a.s. $x\in E$. But for simplicity, we assume $K(x,G)=1\ (x\in E)$.

\begin{definition}[Haar mixture Metropolis kernel]
Let $\Pi(\dif x)=\exp(-U(x))\mu_*(\dif x)$. 
The Haar mixture Metropolis kernel is a Metropolis kernel with proposal kernel
$$
Q_*(x,A)=\int_{g\in G}K(x,\dif g)Q_g(x, A)
$$
with acceptance probability (\ref{eq:acceptance_probability}). 
The Makov kernel $Q_*$ is $\mu_*$-reversible. 
\end{definition}

\begin{example}[Autoregressive kernel]
\label{ex:ar-intro}
Let $h\in [0,2\pi)$.  Consider a Markov kernel
\[
Q(x,\cdot)=\mathcal{N}_d(\cos h~x, (\sin h)^2
I_d)
\]
on $\mathbb{R}^d$.
Simple calculation yields that it is $\mu=\mathcal{N}_d(0, I_d)$-reversible. 

Let  $G=(\mathbb{R}_+,\times)$ and set $(g,x)\mapsto g^{1/2}x$. Then the Haar measure is $\nu(\dif g)~\propto~ g^{-1}\dif g$. A simple calculation yields $\mu_g=\mathcal{N}_d(0, g^{-1}I_d)$ and  
$Q_g(x,\cdot)=\mathcal{N}_d(\cos h~x, g^{-1}(\sin h)^2 I_d)$. Also, 
$\mu_*(\dif x)~\propto~ |x|^{-d}\dif x$ and
$K(x,\dif g)=\mathcal{G}(d/2, |x|^2/2)$. 
We have a closed form (up to a constant) of expression of $Q_*(x,\cdot)$ as follows: 
\begin{align*}
    Q_*(x,\mathrm{d}y)\propto \left[1+\frac{|y-\cos h~x|^2}{(\sin h)^2~|x|^2}\right]^{-d}\dif x. 
\end{align*}
See \cite{kamatani2020} for the detail. 
\end{example}

\subsection{Benefit of using heavy-tailed reference measure}

In this section, we provide a result that explains in part the importance of the heavy-tailed reference measure (Definition \ref{def:metropolis}). More precisely, for uniform ergodicity, we show that $\mu$ must be heavier than the target distribution $\Pi$ so that the density is bounded from above.  
As explained in Section \ref{sec:haar}, Haar mixture can make the reference measure heavier. This explains in part why Haar mixture improves performance. 

Suppose that the measurable space $(E,\mathcal{E})$ is countably generated. By this condition, there exists a small set of positive $\Pi$-probability \cite[see Section 5.2 of ][]{MT}. 
A Markov kernel $P$ is uniformly ergodic if there exists a probability measure $\Pi$ such that
\begin{equation}
    \nonumber
    \sup_{x\in E}\sup_{A\in\mathcal{E}}\left|P^m(x,A)-\Pi(A)\right|\longrightarrow_{m\rightarrow\infty}0. 
\end{equation}
This theorem is a kind of generalisation of Theorem 2.1 of \cite{MT2}. 

\begin{theorem}\label{thm:uniform_ergodicity}
Suppose that $\mu$ is an atomless $\sigma$-finite measure, and $Q(x,\dif y)=q(x,y)\mu(\dif y)$ is a $\mu$-reversible Markov kernel with a jointly measurable density $q(x,y)$. 
Suppose that a probability measure $\Pi$ has a density $\pi(x)$ with respect to $\mu$ such that
$$
\mu\operatorname{-ess\ sup}\pi(x)=+\infty. 
$$
Then the Metropolis kernel $P$ is not uniformly ergodic. 
\end{theorem}

\begin{proof}
Suppose now that $P$ is uniformly ergodic. Then Markov kernel is $\Pi$-irreducible and also aperiodic  \cite[see Theorem 5.4.4 of ][]{MT}. 
Also, there exists $c>0$ and a probability measure $\xi$ such that 
\begin{equation}\label{eq:uniform-ergodicity}
P^M(x,\dif y)\ge c~\xi(\dif y)\ (x\in E)
\end{equation}
for some $M\in\mathbb{N}$  \cite[Theorem 16.0.2 of][]{MT}.  By construction, the Metropolis kernel $P$ satisfies 
$$
P(x,\dif y)\le I(x,\dif y)+Q(x,\dif y)\alpha(x,y)
$$
where $I(x, A)=\delta_x(A)$ and $\alpha(x,y)=\min\{1, \pi(y)/\pi(x)\}$. Observe that $\alpha(x,y)\alpha(y,z)\le \alpha(x,z)$.  
Therefore
\begin{align*}
    P^M(x,\dif y)&\le (I+Q\alpha)^M(x,\dif y)\\
    &=I(x,\dif y)+\sum_{m=1}^M\binom{M}{m}(Q\alpha)^m(x,\dif y)\\
    &\le I(x,\dif y)+\sum_{m=1}^M\binom{M}{m}Q^m(x,\dif y)\alpha(x,y)=:I(x,\dif y)+Q_*(x,\dif y)\alpha(x,y). 
\end{align*}
Observe that $Q_*(x, E)=\sum_{m=1}^M\binom{M}{m}=2^M-1$. 
For each $x\in E$,  $P^M(x,\dif y)$ has a unique atom at the current state $x$ and $P^M(x,\dif z)$ and $P^M(y,\dif z)$ does not share an atom if $x\neq y$. Therefore, $\xi$ does not have an atom, and we have
\begin{equation}
Q_*(x,\dif y)\alpha(x,y)\ge c~\xi(\dif y). 
\label{eq:small_set}
\end{equation}
Since $Q^m(x,\dif y)$ is absolutely continuous with respect to $\mu$ for each $x\in E$, the probability measure $\xi$ is also absolutely continuous and it has a density function $h(x)=\dif\xi/\dif\mu(x)$. Observe that without loss of generality,  we can assume that $h(x)$ is bounded above since (\ref{eq:uniform-ergodicity}) holds for a truncated version of $h$, that is, $P^M(x,\dif y)\ge c_K~\xi^K(\dif y)$ where the probability measure $\xi^K$ is  $\xi^K(\dif y)=\min\{K,h(y)\}\mu(\dif y)/\int \min\{K,h(z)\}\mu(\dif z)$ for $c_K=c~\int \min\{K,h(z)\}\mu(\dif z)$ for some $K>0$ such that $c_K>0$.   

Let $q_*(x,y)$ be the density of $Q_*(x,\dif y)$ with respect to $\mu$.   
%By construction, $\int_Eq_*(x,y)\mu(\dif y)=Q_*(x, E)=2^M-1$. 
The left-hand side of (\ref{eq:small_set}) is 
$$
q_*(x,y)\mu(\dif y)\alpha(x,y)
=q_*(x,y)\Pi(\dif y)\min\left\{\frac{1}{\pi(y)},\frac{1}{\pi(x)}\right\}. 
$$
For any $N>0$, set $N(x)=\{y\in E:q_*(x,y)\le N\}$. 
By assumption, there exists $x_n\in E\ (n=1,2,\ldots)$ such that 
$\pi(x_n)\rightarrow \infty$. Then 
$$
\int_{y\in N(x_n)} q_*(x_n,y)\Pi(\dif y)\min\left\{\frac{1}{\pi(y)},\frac{1}{\pi(x_n)}\right\}\le \frac{N}{\pi(x_n)}\longrightarrow_{n\rightarrow\infty} 0. 
$$
Therefore, by the inequality (\ref{eq:small_set}), 
\begin{equation}\label{eq:uniform-ergodic-convergence}
    \xi(N(x_n))\longrightarrow_{n\rightarrow\infty} 0~\Longrightarrow~\xi(N(x_n)^c)\rightarrow_{n\rightarrow\infty} 1. 
\end{equation}
On the other hand, by Markov's inequality, we have
\begin{align*}
\mu(N(x_n)^c)\le \int_E\frac{q_*(x_n,y)}{N}\mu(\dif y)= \frac{Q_*(x_n, E)}{N}=\frac{2^M-1}{N}
\end{align*}
and hence
\begin{align*}
\xi(N(x_n)^c)=\int_{N(x_n)^c}h(y)\mu(\dif y)\le \|h\|_\infty\mu(N(x_n)^c)\le \|h\|_\infty\frac{2^M-1}{N} 
\end{align*}
where $\|h\|_\infty=\sup_{x\in E}|h(x)|$. 
This contradicts the convergence (\ref{eq:uniform-ergodic-convergence}) when $N$ is greater than $\|h\|_\infty~(2^M-1)$. Therefore, $P$ is not uniformly ergodic. 
\end{proof}

We do not investigate further ergodic properties of our new kernels that will be explained in Section \ref{sec:wmh}. The exponential ergodicity of the new kernels is an interesting topic, but it exceeds the scope of our paper.

\section{Weave kernels as combination of transform-based and Haar mixture kernels}

\subsection{Augmentation and the Haar measure}

In this section, we introduce a new kernel that combines a flipped reversible transform and the Haar mixture Metropolis kernel.  The new kernel is intended to have a locally and globally informed. To this end, we consider a slightly simpler version of Definition \ref{def:mh_augmented}. We call this simpler version Metropolis kernel, not Metropolis--Hastings kernel,  since the acceptance rate depends only on the fraction of target densities.

\begin{definition}[Augmented Metropolis kernel]
\label{def:amk}
Let $\mu$ be a probability measure. 
Let $\Phi:E^2\rightarrow E^2$ be a $(\mu^{\otimes 2},\kappa)$-reversible transform. Let $\Pi(\dif x)=\exp(-U(x))\mu(\dif x)$ be a probability measure on $(E,\mathcal{E})$. The augmented Metropolis kernel is a Metropolis kernel on $(E,\mathcal{E})$ with the proposal kernel
\begin{equation}
    Q(x,A)=\mu(\{w\in E:\Phi(x,w)\in A\times E\})\ (x\in E, A\in\mathcal{E})
    \label{eq:metrpolis_with_augmented_var}
\end{equation}
with acceptance probability (\ref{eq:acceptance_probability}). 
\end{definition}

Based on this Metropolis kernel, we construct a Haar mixture version. %See the next section for $\Pi$-reversibility of the kernel. %Let 
% $$
% \mu_*^{\otimes 2}(\dif x,\dif v)=\int_{g\in G}\nu(\dif g)\mu_g(\dif x)\mu_g(\dif v). 
% $$
% Note that $\mu_*^{\otimes 2}\neq (\mu_*)^{\otimes 2}$ in general.  

\begin{definition}[Augmented Haar--Metropolis kernel]
\label{def:ahm}
Let $\mu$ be a probability measure, and 
let $\Phi:E^2\rightarrow E^2$ be a $(\mu_g^{\otimes  2},\kappa)$-reversible transform for any $g\in G$. Let $\Pi(\dif x)=\exp(-U(x))\mu_*(\dif x)$ be a probability measure on $(E,\mathcal{E})$. The augmented Haar metropolis kernel is a Metropolis kernel on $(E,\mathcal{E})$ with the proposal kernel
\begin{equation}
    Q_*(x,A)=\int_GK(x,\dif g)\mu_g(\{w\in E:\Phi(x,w)\in A\times E\})\ (x\in E, A\in\mathcal{E})
    \label{eq:haar_metrpolis_with_augmented_var}
\end{equation}
with acceptance probability (\ref{eq:acceptance_probability}).
\end{definition}

In Definitions \ref{def:amk} and \ref{def:ahm}, proposal kernel should be measure preserving with respect to a probability measure, say, $\mu$. This assumption is crucial. There are many transforms that are Lebesgue measure preserving. These transforms can be $\mu$-measure preserving if the transform does not change the value of the density (see Proposition \ref{prop:reversible_flow}-4). We call that Lebesgue-measure preserving transform satisfies the \textbf{density preserving condition} if it does not change the value of the density.  However, most gradient-based transforms do not satisfy this condition. An important exception is the bounce transform, whose kernel will form the basis for our new Metropolis kernels.

%\subsection{Reversibility of Haar mixture kernels}

\subsection{Weave-Metropolis kernel}
\label{sec:wmh}

We are in the position to introduce two new kernels, which use a bounce transform, and a circle transform introduced by the following definition.

% At first, we introduce a robust global motion, it was widely used for sampling from the latent Gaussian variables model, which is also very robust to dimension.
\begin{definition}{(\textbf{Circle transform})}\label{def:circular} For $x\in \mathbb{R}^d$, $v\in \mathbb{R}^d$, and $h \in [0,2\pi)$, the circle transform denoted by
$(x^*,v^*)=\Phi_{\Circle}(x,v)$ is defined by 
\begin{align*}
x^*&\leftarrow x\cos h + v \sin h\\
v^*&\leftarrow -x\sin h +v\cos h. 
\end{align*}
More generally, for a $d$-dimensional vector $M$, a circle transform with parameter $M$ is denoted by 
$(x^*,v^*)=\Phi_{\Circle}(x,v\mid M)$ and defined by
\begin{align*}
x^*&\leftarrow M+(x-M)\cos h + (v-M) \sin h\\
v^*&\leftarrow M-(x-M)\sin h + (v-M)\cos h. 
\end{align*}
\end{definition}

% If we consider $h\in [0,2\pi)$ as a parameter, then the circle transform forms a two-dimensional circle with the center at the origin. Also, the $x$ coordinate forms a two-dimensional ellipse with the center at the origin and the intersection of $x$ and $v$. 

The circle transform is $(\mathcal{N}_d(0,\Sigma)^{\otimes 2},\kappa)$-reversible for $\kappa(x,v)=(x,-v)$ and has been used as a proposal kernel of Metropolis--Hastings kernels \citep{MR1723510,MR2444507, murray2010, Bierkens2020TheBS}.  
The transform naturally fits the Gaussian prior  distribution used in many statistical problems. %\kengo{Clarify} The Gaussian-reversible transform naturally fits high-dimensional statistical problems, while it is difficult for the Lebesgue measure-reversible transforms because there is no natural extension of the Lebesgue measure in infinite dimensional space, which is explained in part in \cite{MR3262508}. 

% , they proved that the random-walk Metropolis algorithm's spectral gap is shrinkage with dimension, but the spectral gap of pCN kernel is independent of its dimension. There are lots of applications about the circle transform. For example, In \cite{murray2010}, they proposed the elliptical slice sampler. It combines the circle transform and slice sampler together and explores the target distribution.

% \begin{definition}(\textbf{bounce transform})
% $$
% v\longleftarrow v-2\left \langle v,\frac{\xi(x)}{|\xi(x)|}\right \rangle \frac{\xi(x)}{|\xi(x)|}. 
% $$
% where $v$ is the velocity variable, and $\xi:\mathbb{R}^d\rightarrow \mathbb{R}^d$ is a continuous function. 

% \end{definition}

% $\phi$ is $\nu$-reversible if 
% $$
% \nu(\{x\in A: \phi(x)\in B\})
% =\nu(\{x\in B: \phi(x)\in A\}). 
% $$

Based on the circle transform, we would like to construct an efficient transform that uses local information about the potential energy. It is possible to introduce global information into the circle transform, as in \cite{LAW2014127, Rudolf_2016, CUI2016109}. However, the circle transform is blind to the local information of potential energy. Therefore, we need to combine another transform to fulfill our purpose, i.e., to connect with the local information. 
Hamiltonian flow is commonly used to introduce local information. However, we do not use this strategy because Hamiltonian flow does not satisfy the density preserving condition with respect to the normal distribution. 

In this work, we use the bounce transform to employ the local information of the target distribution, and introduce the Weave transform: 
\begin{equation}
\label{map:flow}
\phi_h(z)=(\Phi_{\mathrm{circle}}\circ \Phi_{\mathrm{bounce}}\circ \Phi_{\mathrm{circle}})(z). 
\end{equation}
The transform is similar to the Hug kernel in Example \ref{ex:hug}. The only difference is that it replaces the shift transform with the circle transform. Thanks to this difference, the Weave transform always preserves the distance from the origin. 
This transform satisfies the density preserving condition and also keeps the value of the potential energy approximately constant. The latter assertion is described in Section \ref{sec:limit}.  
The \textbf{Weave-Metropolis kernel} is the Metropolis kernel using the Weave transform and is considered a discrete-time version of the boomerang sampler proposed in \cite{Bierkens2020TheBS}. 
\begin{definition}[Weave-Metropolis kernel]
Let $L\in\mathbb{N}$ and $h>0$, and let $\kappa(x,v)=(x,-v)$. Let $\Pi(\dif x)=\exp(-U(x))\mu(\dif x)$ be a probability measure on $\mathbb{R}^d$ with $\mu=\mathcal{N}_d(0,I_d)$. The Weave-Metropolis kernel is an augmented Metropolis kernel with $\Phi=(\phi_h)^L$. 
\end{definition}

We also introduce the Haar mixture version of the Weave-Metropolis kernel. We will see that the Haar mixture version, the \textbf{Haar-Weave-Metropolis } kernel, performs much better for many target probability distributions. 

\begin{definition}[Haar-Weave-Metropolis kernel]
Let $L\in\mathbb{N}$ and $h>0$, $\kappa(x,v)=(x,-v)$, and let $\mu=\mathcal{N}_d(0, I_d)$. Let $\Pi(\dif x)=\exp(-U(x))\mu_*(\dif x)$ be a probability measure on $\mathbb{R}^d$ with $\mu_*(\dif x)=|x|^{-d}\dif x$. The Haar-Weave-Metropolis kernel is an augmented Haar--Metropolis kernel on $\mathbb{R}^d$ with $\Phi=(\phi_h)^L$. 
In this case,  $K(x,\dif g)=\mathcal{G}(d/2, |x|^2/2)$. 
\end{definition}

The step-by-step description of the algorithms using Weave-Metropolis and Haar-Weave-Metropolis kernels are described in Algorithms \ref{al:wm} and \ref{al:hwm}.

\subsection{Reversibility of Metropolis kernels}

%\subsection{Properties of the flipped reversible transform}
\label{sec:reversibility}

In this section, we further investigate the flipped reversible transforms. 
We would like to remind the reader here that, 
if the absolute value of the Jacobian determinant of $\psi:\mathbb{R}^d\rightarrow\mathbb{R}^d$ is $1$ and one-to-one, then $\psi$ is Lebesgue measure preserving. It follows that, 
the bounce transform (Definition \ref{def:bounce}) and the circle transform (Definition \ref{def:circular}) are Lebesgue measure preserving.

\begin{proposition}\label{prop:reversible_flow}
Let $\mu$ and $\nu$ be $\sigma$-finite measures on $(E,\mathcal{E})$ such that $\nu\ll \mu$. Let $p(x)=\dif\nu/\dif \mu(x)$. 
Let $\kappa:E\rightarrow E$ be a $\mu$-reversible transform such that $\kappa\circ\kappa(x)=x$. 
\begin{enumerate}
    \item 
If the $\mu$-measure preserving transform $\psi$ satisfies the condition $\kappa\circ\psi\circ\kappa\circ\psi(x)=x$ in $\mu$-a.s., then $\psi$ is a $(\mu,\kappa)$-reversible transform. 
\item If $\psi$ is $(\mu,\kappa)$-reversible, then $\psi^L$ is $(\mu,\kappa)$-reversible for $L\in\mathbb{N}$. 
\item If $\psi_1, \psi_2$ are $(\mu,\kappa)$-reversible transforms, then $\psi_1\circ\psi_2\circ\psi_1$ is also $(\mu,\kappa)$-reversible. 
\item If $\psi$ is $(\mu,\kappa)$-reversible and if $p(x)=p(\kappa(x))$ and $p(x)=p(\psi(x))$ in $\mu$-a.s., then $\psi$ is also $(\nu,\kappa)$-reversible. 
\end{enumerate}
\label{lemma:leb-rev}
\end{proposition}

\begin{proof}
\begin{enumerate}
    \item By substituting $\kappa\circ\psi$ for $\psi$, it suffices to show that $\psi\circ\psi(x)=x$ implies the $\mu$-reversibility of $\psi$ when $\psi$ is  $\mu$-measure preserving. However, for any $A, B\in\mathcal{E}$ we have.
\begin{align*}
\mu(\{x\in A:\psi(x)\in B\})
&=\mu(\{\psi\circ\psi (x)\in A:\psi(x)\in B\})\\
&=\mu(\{\psi (x)\in A:x\in B\}). 
\end{align*}
\item  For any transform $\phi$ with a $(\mu,\kappa)$-reversible transform $\psi$, the following holds: 
\begin{align*}
\mu(\{x\in A:\kappa\circ\phi\circ\psi(x)\in B\})
&=\mu(\{x\in A:\kappa\circ\phi\circ\kappa\circ\kappa\circ\psi(x)\in B\})\\\
&=\mu(\{\kappa\circ\psi(x)\in A:\kappa\circ\phi\circ\kappa(x)\in B\})\\
&=\mu(\{\kappa\circ\psi\circ\kappa(x)\in A:\kappa\circ\phi(x)\in B\}). 
\end{align*}
Applying this equation sequentially, we obtain
\begin{align*}
\mu(\{x\in A:\kappa\circ\psi^L(x)\in B\})&=
\mu(\{\kappa\circ\psi^L\circ\kappa(x)\in A:\kappa(x)\in B\})\\\
&=
\mu(\{\kappa\circ\psi^{L}(x)\in A:x\in B\}). 
\end{align*}
\item 
A similar argument as above proves the assertion. 
\item By construction we have
\begin{align*}
\nu(\{x\in A:\kappa\circ\psi(x)\in B\})
&=\int_E1(\{x\in A:\kappa\circ\psi(x)\in B\})p(x)\mu(\dif x)\\\
&=\int_E1(\{\kappa\circ\psi(x)\in A:x\in B\})p(\kappa\circ\psi(x))\mu(\dif x)\\\
&=\int_E1(\{\kappa\circ\psi(x)\in A:x\in B\})p(x)\mu(\dif x)\\
&=\nu(\{\kappa\circ\psi(x)\in A:x\in B\}). 
\end{align*}
\end{enumerate}
\end{proof}

\begin{proposition}\label{prop:metropolis_hastings_augmented_var}
Augmented Metropolis--Hastings kernel $P$ in Definition  \ref{def:mh_augmented} is $\Pi$-reversible. 
\end{proposition}

\begin{proof}
It suffices to show that $\widehat{P}$ is $\Pi$-reversible, where
$$
\widehat{P}(x,A):=\int_{v\in E} 1(\{\Phi(x,v)\in A\times E\})~\lambda(\dif v)
    \alpha((x,v),\Phi(x,v)). 
$$
By construction we have $\Pi(\dif x)\lambda(\dif v)=e^{-H(z)}\mu^{\otimes 2}(\dif z)$ and 
\begin{align*}
    \int_A\Pi(\dif x)\widehat{P}(x,B)
    &=\int_{(x,v)\in A\times E}\Pi(\dif x) 1(\{\Phi(x,v)\in B\times E\})~\lambda(\dif v)
    \alpha((x,v),\Phi(x,v))\\
    &=\int\mu^{\otimes 2}(\dif z)\beta(z,\Phi(z))
\end{align*}
where 
$$
\beta(z,z')=1(\{z\in A\times E,\ z'\in B\times E\})~\min\{e^{-H(z)},e^{-H(z')}\}. 
$$
Since $K(v)=K(-v)$, we have $\beta(z,\Phi(z))=\beta(z,\kappa\circ\Phi(z))$. On the other hand, by $(\mu^{\otimes 2},\kappa)$-reversibility of $\Phi$, 
\begin{align*}
    \int_A\Pi(\dif x)\widehat{P}(x,B)
    &=\int\mu^{\otimes 2}(\dif z)\beta(z,\kappa\circ\Phi(z))\\
    &=\int\mu^{\otimes 2}(\dif z)\beta(\kappa\circ\Phi(z),z)\\
    &=\int\mu^{\otimes 2}(\dif z)\beta(\Phi(z),z)
    =\int_B\Pi(\dif x)\widehat{P}(x,A). 
\end{align*}
Thus the claim follows. 
\end{proof}

\begin{remark}
Let $\phi=\Phi_{\Bounce}$ or $\phi=\Phi_{\Circle}$. 
For $\kappa(x,v)=(x,-v)$,  we have $\kappa\circ\phi\circ\kappa\circ\phi(x,v)=(x,v)$. 
Since $\phi$ is Lebesgue measure preserving map, it is $(\mathrm{Leb},\kappa)$-reversible flow. Also, since $|z|=|\kappa(z)|=|\phi(z)|$, it is 
$(\mu_g^{\otimes 2},\kappa)$-reversible for $\mu_g=\mathcal{N}_d(0,g^{-1} I_d)$. Finally, $\phi_h$ in (\ref{map:flow}) is $(\mu_g^{\otimes 2},\kappa)$-reversible. 
\end{remark}

\begin{proposition}
The proposal kernels $Q$ and $Q_*$ defined in Definitions \ref{def:amk} and \ref{def:ahm} are $\mu$ and $\mu_*$-reversible respectively. In particular, the augmented Metropolis and augmented Haar--Metropolis  kernels are $\Pi$-reversible. 
\end{proposition}

\begin{proof}
% For $A, B\in\mathcal{E}$, we have
% \begin{align*}
%     \int_A\mu(\dif x)Q(x,B)&=\int_A\mu(\dif x)\mu(\{v\in E:\Phi(x,v)\in B\times E\})\\
%     &=\mu^{\otimes 2}(\{(x,v)\in A\times E,\ \Phi(x,v)\in B\times E\}). 
% \end{align*}
% Observe that $A\times E$ and $B\times E$ are $\kappa$-invariant sets. Thus above probability is
% \begin{align*}
%     \mu^{\otimes 2}(\{(x,v)\in A\times E,\ \kappa\circ\Phi(x,v)\in B\times E\})
%     &= 
%     \mu^{\otimes 2}(\{\kappa\circ\Phi(x,v)\in A\times E,\ (x,v)\in B\times E\})\\
%     &= 
%     \mu^{\otimes 2}(\{\Phi(x,v)\in A\times E,\ (x,v)\in B\times E\})
% \end{align*}
% thanks to $(\mu^{\otimes 2},\kappa)$-reversibility of $\Phi$. Thus $Q$ is $\mu$-reversible since above probability is 
% \begin{align*}
%   \int_B\mu(\dif x)\mu(\{v\in E:\Phi(x,v)\in A\times E\})=\int_B\mu(\dif x)Q(x,A). 
% \end{align*} Similarly, 
The proof of $Q$ is essentially the same as that of Proposition \ref{prop:metropolis_hastings_augmented_var}. For $Q_*$, by the identity $\mu_*(\dif x)K(x,\dif g)=\nu(\dif g)\mu_g(\dif x)$, we have
\begin{align*}
    \int_A\mu_*(\dif x)Q_*(x, B)&=
    \int_A\mu_*(\dif x)\int_GK(x,\dif g)\mu_g(\{v\in E:\Phi(x,v)\in B\times E\})\\\
    &=
    \int_G\nu(\dif g)\int_A\mu_g(\dif x)\mu_g(\{v\in E:\Phi(x,v)\in B\times E\})\\
    &=
    \int_G\nu(\dif g)~\mu_g^{\otimes 2}(\{(x,v)\in A\times E,\ \Phi(x,v)\in B\times E\}). 
\end{align*}
The rest of the proof is also similar, since $\Phi$ is $(\mu_g^{\otimes 2},\kappa)$-reversible. The assertion for the Metropolis kernel is obvious. 
\end{proof}

\section{Limit of the Weave transform}

% \subsection{Limit Hamiltonian process on the sphere}
\label{sec:limit}

We will take a closer look at the Weave transform $\phi_h$ defined in (\ref{map:flow}). More precisely, we provide a short-time expansion of $\phi_h$ together with the limit process induced by the transform. These results give a good insight into the behaviour of the Weave-Metropolis kernel, which is different from similar Markov kernels such as the Hamiltonian Monte Carlo kernel.

Let $\xi=\nabla U$. For $\alpha\in\mathbb{R}^d$, let
$$
\mathcal{M}_\alpha=\{x\in\mathbb{R}^d:U(x)=\alpha\}
$$
be the level set of $U$. We will assume that $\mathcal{N}\cap \mathcal{M}_\alpha=\emptyset$ where $\mathcal{N}=\{x:\xi(x)=0\}$. In this case $\mathcal{M}_\alpha$ is a $d-1$-dimensional regular submanifold. %Theorem 9.11 of An introduction to Manifolds 
For $\epsilon>0$ we introduce an $\epsilon$-perturbation of $\mathcal{M}_\alpha$ by
$$
\mathcal{M}_\alpha^\epsilon=\{x\in\mathbb{R}^d:\exists y\ \mathrm{s.t.}\ |x-y|\le \epsilon,\ |U(y)-\alpha|\le \epsilon\}. 
$$
Let $B_r=\{x\in\mathbb{R}^d:|x|\le r\}$.  
Consider a projection of $v$ to the tangential space of $\xi(x)$: 
$$
P(x)v=(I-\overline{\xi}(x)\overline{\xi}(x)^{\top})v
$$
for $v\in\mathbb{R}^d$ for $x\notin \mathcal{N}$.  
Moreover, let $Q(x)=I-P(x)=\overline{\xi}(x)\overline{\xi}(x)^{\top}$. From this notation we have 
\begin{equation}\label{eq:p_q_r}
R=P-Q,\ 
I=P+Q \quad\leadsto\quad R-I=-2Q,\ R+I=2P
\end{equation}
where $R(x)$ is as in Definition \ref{def:bounce}. 
Also, 
\begin{equation}
P(x)\overline{\xi}(x)=0,\ 
Q(x)\overline{\xi}(x)=\overline{\xi}(x).
\label{eq:p_and_q}
\end{equation}
Thus $P(x)Q(x)=Q(x)P(x)=0$. Observe that 
if $\xi$ is twice differentiable and if $x\notin\mathcal{N}$, then 
\begin{equation}
\label{eq:expression_of_dxi}
\partial\overline{\xi}(x)[u,v]=\sum_{i,j}\partial_{x_j}\overline{\xi}_i(x)u_iv_j=\frac{\partial\xi(x)[P(x)u,v]}{|\xi(x)|}. 
\end{equation}

% The derivatives of $\overline{\xi}(x)$ can be written by using $P(x)$. Therefore, first we provide derivatives of $P(x)$ by using derivatives of $\overline{\xi}(x)$. We omit the proof of the following lemma since it is straightforward. 
First we show that the transform $\phi_h$ does not change the value of $U(x)$ so much. 
Let $p_X(x,v)=x$. See \cite{ludkin2019hug} for the same analysis for the Hug transform. 

\begin{lemma}\label{lem:deviation_of_u}
Let $\alpha\in\mathbb{R}, r, T>0$. 
Suppose that $\xi$ is continuously differentiable. 
Suppose that  $z=(x,y)\in\mathbb{R}^{d}\times\mathbb{R}^d$, 
satisfies $|z|\le r$. Then
\begin{align*}
    |U(p_X(\phi_h(z)))-U(x)|\le h^2~C(r)
\end{align*}
for any $h>0$, where $C(r)=r\sup_{x\in B_r}|\xi(x)| +r^2\sup_{x\in B_r}|\partial\xi(x)|$. 
\end{lemma}

\begin{proof}
Let us introduce the temporary notation $x_0, x_{i}, v_{i}, v_{i}^*\ (i=\pm 1)$ defined by the following chain relation: 
\begin{align*}
z=(x_{-1},v_{-1})\xrightarrow{\Phi_{\mathrm{circle}}} 
(x_0, v_{-1}^*)\xrightarrow{\Phi_{\mathrm{bounce}}}(x_0, v_{1}^*)\xrightarrow{\Phi_{\mathrm{circle}}}
(x_{1},v_{1})=\phi_h(z). 
\end{align*}
We can rewrite the both ends of the chain $(x_{i},v_{i})\ (i=\pm 1)$ with the intermediate state $x_0$ and the velocities $v_{i}^*\ (i=\pm 1)$:  
\begin{equation}
    \begin{split}
    \begin{pmatrix}
      x_{-1}\\
      v_{-1}
    \end{pmatrix}
    &=\Phi_{\mathrm{circle}}^{-1}(x_0,v_{-1}^*)=
    \begin{pmatrix}x_0\cos h-v_{-1}^*\sin h\\ x_0\sin h+v_{-1}^*\cos h
    \end{pmatrix},\\ 
    \begin{pmatrix}
      x_{1}\\
      v_{1}
    \end{pmatrix}&=
    \Phi_{\mathrm{circle}}(x_0, v_1^*)=
    \begin{pmatrix}x_0\cos h+v_1^*\sin h\\ -x_0\sin h+v_1^*\cos h
    \end{pmatrix}. 
    \label{eq:circle_expression}
    \end{split}
\end{equation}
Since $v_1^*=R(x_0)v_{-1}^*$, from (\ref{eq:p_q_r}), 
the following simple relations are obtained for the two intermediate velocities:
\begin{equation}
    \begin{split}
    v_1^*-v_{-1}^*&=(R(x_0)-I)v_{-1}^*=-2Q(x_0)v_{-1}^*\\
    v_1^*+v_{-1}^*&=(R(x_0)+I)v_{-1}^*=2P(x_0)v_{-1}^*.
    \label{eq:identity_v}
    \end{split}
\end{equation}
We need to estimate the difference
\begin{equation}
U(x_1)-U(x_{-1})=\sum_{i=\pm 1}i(U(x_i)-U(x_0))
=\sum_{i=\pm 1}i(U(x_0\cos h+iv_{i}^*\sin h)-U(x_0)). 
\label{eq:expansion_u}
\end{equation}
We assume that $x_0\notin\mathcal{N}$ since if $x_0\in\mathcal{N}$, then $v_{-1}^*=-v_1^*$ and hence $x_{-1}=x_1$, i.e. $U(x_{-1})=U(x_1)$. Since $\xi=\nabla U$, the right-hand side of the equation (\ref{eq:expansion_u}) is 
\begin{align*}
    \sum_{i=\pm 1}i(U(x_i)-U(x_0))
    &=\left\{-\sum_{i=\pm 1}i\int_0^h\xi(x_0\cos \theta+iv_i^*\sin\theta)^{\top}x_0~\sin\theta~\dif\theta\right\}\\
    &\quad~+\left\{\sum_{i=\pm 1}\int_0^h\xi(x_0\cos\theta+iv_i^*\sin\theta)^{\top}v_i^*\cos\theta~\dif\theta\right\}. 
\end{align*}
The absolute value of the first term in the right-hand side is dominated above by $ \sup_{x\in B_r}|\xi(x)|~rh^2$ since $x_0\cos\theta+iv_i^*\sin\theta\in B_r$ and $|\sin\theta|\le |\theta|$.   By (\ref{eq:p_and_q}) and (\ref{eq:identity_v}), we have $\sum_{i=\pm 1}\xi(x)v_i^*=0$. Thus the second term is 
\begin{align*}
    \sum_{i=\pm 1}\int_0^h\left\{\xi(x_0\cos\theta+iv_i^*\sin\theta)-\xi(x_0)\right\}^{\top}v_i^*\cos\theta\dif\theta
\end{align*}
whose absolute value is dominated above by $\sup_{x\in B_r}|\partial\xi(x)| r^2h^2$. Thus the claim follows. 
\end{proof}

We focus on a projection 
$$
\mathbf{p}(z)=\begin{pmatrix}
      x\\
      P(x)v
    \end{pmatrix}. 
$$
We have the following decomposition of $z$ using $\mathbf{p}(z)$:  
\begin{align*}
    z=\begin{pmatrix}
      0\\
      Q(x)v
    \end{pmatrix}
    +
    \begin{pmatrix}
      x\\
      P(x)v
    \end{pmatrix}=\begin{pmatrix}
      0\\
      Q(x)v
    \end{pmatrix}
    +
    \mathbf{p}(z). 
\end{align*}
The first part in the right-hand side corresponds the sign flip, and it is the fast move of the sequence $\{(\phi_h)^n(z):n=0,1,\ldots\}$. The second part is tangential to the sign flip, and it evolves slowly. We are interested into the second part since the first move will be canceled out in the long run. A short term expansion of the transform $\mathbf{p}(\phi_h(z))$  is composed by 
\begin{equation}
\tildeb{a}(x,v)=2\begin{pmatrix}P(x)v\\
-P(x)x
\end{pmatrix},
\label{eq:a}
\end{equation}
and 
\begin{equation}
    \tildeb{b}(x,v)=\partial P(x)[\cdot, v,v]+\partial P(x)[\cdot, R(x)v,R(x)v]\in\mathbb{R}^d
    \label{eq:b}
\end{equation}
where $\partial P$ is defined on $\mathcal{N}^c$ as 
\begin{equation}
\label{eq:derivative_p}
    \partial P(x)[u,v,w]=\sum_{i,j,k}u_i\partial_{x_k} P_{ij}(x)v_jw_k=-(\overline{\xi}(x)^{\top}v)\partial\overline{\xi}(x)[u,w]-(\overline{\xi}(x)^{\top}u)\partial\overline{\xi}(x)[v,w]. 
\end{equation}
%\song{The definition of $\partial\overline{\xi}(x)[v,w]$ doesn't appear until (\ref{eq:expression_of_dxi}), we should move the definition of $\partial\overline{\xi}(x)[v,w]$ to here.}
Note here that $|\tildeb{a}(z)|\le 2|z|$ since $|P(x)v|\le |v|$ and $|P(x)x|\le |x|$. 

\begin{proposition}\label{prop:expansion_of_h}
Let $\alpha\in\mathbb{R}$, $r>0$ and $v\in\mathbb{R}^d$. 
Suppose that $\xi$ is twice continuously differentiable. 
Suppose that $\mathcal{M}_\alpha\cap\mathcal{N}^c\neq\emptyset$.  Choose $\epsilon>0$ such that $\mathcal{M}_\alpha^\epsilon\subset\mathcal{N}^c$. 
Then there is a constant $C(\epsilon, \alpha, r)$ such that 
for $z=(x,v)$ with $|z|\le r$ and $x\in \mathcal{M}_\alpha^{\epsilon/2}$, 
\begin{equation}\label{eq:z_expansion}
    \left|\mathbf{p}(\phi_h(z))-\mathbf{p}(z)-h~\left(\tildeb{a}+\begin{pmatrix}0\\
    \tildeb{b}\end{pmatrix}\right)(z)\right|
    \le h^2~C(\epsilon, \alpha, r)
\end{equation}
 for $0<h< \min\{r^{-1}\epsilon,(\epsilon/2C(r))^{1/2}\}$, where $C(r)$ is defined in Lemma \ref{lem:deviation_of_u}. 
\end{proposition}

\begin{proof}
For a function $f(x)$, we use a generic notation $f(x)=R(\epsilon,\alpha, r)$ if $\sup_{x\in B_r\cap\mathcal{M}_\alpha^\epsilon}|f(x)|<\infty$. 
By Lemma \ref{lem:deviation_of_u}, with the same temporary notation, for $0\le\theta\le h$, we have
$$
|U(x_1)-U(x_{-1})|\le h^2C(r)\le \epsilon/2 \quad\leadsto\quad |U(x_{\pm 1})-\alpha|\le \epsilon
$$
by the triangle inequality, since  $x_{-1}\in\mathcal{M}_\alpha^{\epsilon/2}$. On the other hand, since $x_{\pm 1}=x_0\cos h\pm v_{\pm 1}^*\sin h$ by (\ref{eq:circle_expression}), we have
$$|x_0\cos\theta\pm v_{\pm 1}^*\sin\theta-x_{\pm 1}
|\le r\theta\le \epsilon. 
$$
By the two inequalities, we have
\begin{equation}
    x_0\cos \theta\pm v^*_{\pm 1}\sin \theta %\song{(x_0,x_1)?}
    \in\mathcal{M}_\alpha^\epsilon\ (0\le \theta\le h). 
    \label{eq:intermediate}
\end{equation} 
We now proceed to show (\ref{eq:z_expansion}). 
We have
\begin{align}
\begin{pmatrix}
    x_1-x_{-1}\\
    P(x_0)(v_1-v_{-1})
\end{pmatrix}
&\quad\stackrel{(\ref{eq:circle_expression})}{=}
\begin{pmatrix}
    (v_1^*+v_{-1}^*)\sin h\\
    -2P(x_0)x_0\sin h+P(x_0)(v_1^*-v_{-1}^*)\cos h
\end{pmatrix}
\nonumber
\\
&\stackrel{(\ref{eq:p_and_q}),\ (\ref{eq:identity_v})}{=}
\begin{pmatrix}
  ~~2P(x_0)~v_{-1}^*~\sin h\\
  -2P(x_0)~x_0~~~\sin h
\end{pmatrix}
\nonumber
\\
&\quad\stackrel{(\ref{eq:a})}{=}\tildeb{a}(x_0, v_{-1}^*)\sin h. 
    \label{eq:identity_a}
\end{align}
Let 
\begin{align*}
\widehat{b}(x_0,v_{-1}^*)&=h^{-1}\left\{(P(x_1)-P(x_0))v_1-
(P(x_{-1})-P(x_0))v_{-1}\right\}. 
\end{align*}
From (\ref{eq:identity_a}), we obtain 
\begin{align*}
    \mathbf{p}(\phi_h(z))-\mathbf{p}(z)&=
    \begin{pmatrix}
      x_1-x_{-1}\\
      P(x_1)v_1-P(x_{-1})v_{-1}
    \end{pmatrix}
    =
    \tildeb{a}(x_0, v_{-1}^*)
    \sin h
    +
    \begin{pmatrix}
      0\\
      \widehat{b}(x_0,v_{-1}^*)
    \end{pmatrix}~h. 
\end{align*}
Thus it is sufficient to show two inequalities  
\begin{equation}
\begin{split}
&|\tildeb{a}(x_0, v_{-1}^*)\sin h-\tildeb{a}(z) h|\le h^2 R(\epsilon, \alpha, r),\\ 
&|\widehat{b}(x_0,v_{-1}^*)-\tildeb{b}(z)|\le h R(\epsilon, \alpha, r). 
\end{split}
\label{eq:difference}
\end{equation}
\begin{enumerate}
    \item[(a)] 
Since $|\tildeb{a}(z)|\le 2|z|\le 2r$, we have
$$
|\tildeb{a}(x_0, v_{-1}^*)\sin h-\tildeb{a}(x_0, v_{-1}^*)h|=|\tildeb{a}(x_0, v_{-1}^*)|~|\sin h-h|\le 2r\frac{h^3}{3!}. 
$$
Also, by (\ref{eq:circle_expression}) and (\ref{eq:intermediate}), we have
\begin{align*}
\left|\tildeb{a}(x_0, v_{-1}^*)-
\tildeb{a}(x_{-1}, v_{-1})\right|\le rh~\sup_{x\in B_r\cap\mathcal{M}_\alpha^\epsilon, v\in B_r}|\tildeb{a}'(x,v)|=h R(\epsilon,\alpha, r). 
\end{align*}
Thus the first part of (\ref{eq:difference}) follows since $z=(x_{-1},v_{-1})$. 
\item[(b)] 
By (\ref{eq:circle_expression}), we have
\begin{align*}
\widehat{b}(x_0,v_{-1}^*)&=h^{-1}\sum_{i=\pm 1}i(P(x_i)-P(x_0))v_i\\
&=h^{-1}\sum_{i=\pm 1}(P(x_0\cos h+iv_{i}^*\sin h)-P(x_0))(-x_0\sin h+iv_{i}^*\cos h)\\
&=
h^{-1}\sum_{i=\pm 1}\int_0^h\partial P(x_0\cos \theta+iv_{i}^*\sin \theta)[\cdot,(-x_0\sin \theta+iv_{i}^*\cos \theta),(-x_0\sin \theta+iv_{i}^*\cos \theta)]\dif\theta\\
&\quad -h^{-1}\sum_{i=\pm 1}\int_0^h(P(x_0\cos \theta+iv_{i}^*\sin \theta)-P(x_0))(x_0\cos \theta+iv_{i}^*\sin \theta)\dif\theta. 
\end{align*}
The second term in the right-hand side is on the order of $h$ since $P$ is differentiable on $B_r\cap\mathcal{M}_\alpha^\epsilon$. 
Similarly, since $P$ is twice differentiable on $B_r\cap\mathcal{M}_\alpha^\epsilon$, the first term is 
\begin{align*}
    \sum_{i=\pm 1}~\partial P(x_0)[\cdot, v_{i}^*,v_{i}^*]
    +h~R(\epsilon,\alpha, r)
    &=\tildeb{b}(x_0,v_{-1}^*)
    +h~R(\epsilon,\alpha, r). 
\end{align*}
On the other hand, by Taylor's expansion, we have (recall that $z=(x_{-1},v_{-1})$)
\begin{align*}
  \left|\tildeb{b}(x_{-1},v_{-1})-\tildeb{b}(x_0,v_{-1}^*)\right|\le h\sup_{x\in B_r\cap \mathcal{M}_\alpha^\epsilon, v\in B_r}|\tildeb{b}'(x,v)|=h~R(\epsilon,\alpha,r). 
\end{align*}
Thus the second part of (\ref{eq:difference}) follows. 
\end{enumerate}
\end{proof}

We study limit behaviour of the sequence $\{(\phi_h)^n(z):n=0,1,2,\ldots\}$ as $h\rightarrow 0$. 
Let
$z^h(t)=(x^h(t),v^h(t))=(\phi_h)^n(z)$ if $nh\le t<(n+1)h$ for $n=0,1,\ldots$ be the continuous process version of the sequence $\{(\phi_h)^n(z):n=0,1,\ldots\}$. We also define $w^h(t)=P(x^h(t))v^h(t)$ that is the tangential component of $v^h(t)$ to $\overline{\xi}(x^h(t))$. 

\begin{theorem}\label{thm:limit_process}
Suppose that $\xi$ is thrice differentiable and $\inf_{|x|\rightarrow\infty}U(x)=+\infty$. 
Let $\alpha\in\mathbb{R}$ such that $\mathcal{M}_\alpha\subset\mathcal{N}^c$ and $\mathcal{M}_\alpha\neq \emptyset$.
Choose $x\in \mathcal{M}_\alpha$ $v\in\mathbb{R}^d$. 
For $T>0$, consider the process $(x^h(t),w^h(t))$ such that  $z^h(0)=(x,v)$. 
Let $r=|z^h(0)|$. Then there exists $H>0$ such that for $0<h<H$, we have
$$
e_T:=\sup_{0\le t\le T}|\mathbf{p}(z^h(t))-\zeta(t)|\le h~C(\epsilon,\alpha,r, T)
$$
for some $C(\epsilon,\alpha,r, T)>0$
where $\zeta(t)=(x(t),w(t))$ is the solution of 
\begin{equation}
\zeta'(t)=\mathbf{a}(\zeta(t))+\begin{pmatrix}
  0\\ \mathbf{b}(\zeta(t))
\end{pmatrix},\  \zeta(0)=(x(0),w(0))=(x,P(x)v)
\label{eq:limit_equation}
\end{equation}
where 
\begin{equation}
\begin{split}
\mathbf{a}(x,w)&=2(w, -P(x)x)^{\top}, \\
\mathbf{b}(x,w)&=-2(r^2-|w|^2-|x|^2)\partial\overline{\xi}(x)[\cdot,\overline{\xi}(x)]-2
\partial\overline{\xi}(x)[w,w]\overline{\xi}(x).
\end{split}
\label{eq:expression_b}
\end{equation}
\label{pro::3.1}
\end{theorem}

\begin{proof}
First, we consider the existence and uniqueness of the limit process $\zeta(t)$.  Since $\inf_{|x|\rightarrow\infty}U(x)=+\infty$, the level set $\mathcal{M}_\alpha$ is a compact set. Since $\xi$ is thrice differentiable, the functions $a$ and $b$ are differentiable and the derivatives are bounded on $\mathcal{M}_\alpha^\epsilon\times K$ where $K$ is any compact set of $\mathbb{R}^d$. Therefore, by the Picard--Lindel\"of theorem, there exists a unique solution $\zeta(t)=(x(t),w(t))\ (0\le t\le h_1)$ for some $h_1>0$.  By Remark \ref{rem:property_of_the_limit}, $x(h)\in\mathcal{M}_\alpha$. Therefore, by the same argument replacing $t=0$ by $t=h_1$, there exists a unique solution $\zeta(t)=(x(t),w(t))\ (0\le t\le h_2)$ for some $h_2>h_1$. By iterating this argument, we obtain the unique solution $\zeta(t)\ (0\le t\le T)$.

By (\ref{eq:derivative_p}) and with $R(x)\overline{\xi}(x)=-\overline{\xi}(x)$, the vector-valued function $\tildeb{b}$ can be written as
\begin{align}
w^{\top}\tildeb{b}(x,v)
&\stackrel{(\ref{eq:b})}{=}\partial P(x)[w,v,v]+\partial P[w,R(x)v,R(x)v]
\nonumber
\\
&\stackrel{(\ref{eq:derivative_p})}{=}
-(\overline{\xi}(x)^{\top}v)\partial\overline{\xi}(x)[w,v]-(\overline{\xi}(x)^{\top}w)\partial\overline{\xi}(x)[v,v]
\nonumber
\\
&\quad
+(\overline{\xi}(x)^{\top}v)\partial\overline{\xi}(x)[w,R(x)v]-(\overline{\xi}(x)^{\top}w)\partial\overline{\xi}(x)[v,R(x)v]
\nonumber
\\
&\stackrel{(\ref{eq:p_q_r})}{=}-2(\overline{\xi}(x)^{\top}v)^2\partial\overline{\xi}(x)[w,\overline{\xi}(x)]-2
(\overline{\xi}(x)^{\top}w)\partial\overline{\xi}(x)[v,P(x)v],
\label{eq:expression_B}
\end{align}
where we used the fact that 
$\partial\overline{\xi}[v,\cdot]=\partial\overline{\xi}[R(x)v,\cdot]$ by (\ref{eq:p_and_q}) and (\ref{eq:expression_of_dxi}). 
Let 
\begin{align*}
\tildeb{c}(x,v)&=\tildeb{a}(x,v)+(0,\tildeb{b}(x,v))^{\top},\\
\mathbf{c}(x,w)&=\mathbf{a}(x,w)+(0,\mathbf{b}(x,w))^{\top}. 
\end{align*}  
Observe that $a(z)=\mathbf{a}(\mathbf{p}(z))$. Also, 
 if $x^2+v^2=r^2$, then 
$r^2-|P(x)v|^2-|x|^2=|v|^2-|P(x)v|^2=(\overline{\xi}(x)^{\top}v)^2$ and hence 
\begin{align*}
    \mathbf{b}(\mathbf{p}(z))
    &\stackrel{(\ref{eq:expression_b})}{=}
    -2(r^2-|P(x)v|^2-|x|^2)~\partial\overline{\xi}(x)[\cdot,\overline{\xi}(x)]-2\partial\overline{\xi}(x)[P(x)v,P(x)v]\overline{\xi}(x)\\
    &\stackrel{(\ref{eq:p_and_q})}{=}
    -2(\overline{\xi}(x)^{\top}v)^2~\partial\overline{\xi}(x)[\cdot,\overline{\xi}(x)]-2\partial\overline{\xi}(x)[v,P(x)v]\overline{\xi}(x)\\
    &\stackrel{(\ref{eq:expression_B})}{=}b(z). 
\end{align*}
Thus
\begin{equation}
\label{eq:c_identity}
    \tildeb{c}(z)=\mathbf{c}(\mathbf{p}(z)). 
\end{equation}
Let $N=[T/h]$. 
The process $z^h$ always remains on a sphere, i.e. there exists $r>0$ such that 
$$
|z^h(t)|^2=r^2\ (0\le t\le T). 
$$
We choose $\epsilon>0$ so that $\mathcal{M}_\alpha^\epsilon\subset\mathcal{N}^c$ and set 
$$
H =  C(r)^{-1}T^{-1}(\epsilon/2). 
$$
From Lemma \ref{lem:deviation_of_u}, for $n=1,\ldots, N$ and for $0<h<H$, 
$$
|U(x^h(nh)-U(x^h((n-1)h))|\le h^2 C(r)\quad\leadsto\quad |U(x^h(nh))-U(x)|\le \frac{\epsilon}{2}\quad\leadsto\quad x^h(nh)\in\mathcal{M}_\alpha^{\epsilon/2}. 
$$
The inequality (\ref{eq:z_expansion}) holds for every $z^h(nh)$ by Proposition \ref{prop:expansion_of_h}. Thus, for 
$$
\delta^h(t):=\mathbf{p}(z^h(t))-\mathbf{p}(z^h(0))-\int_0^t \tildeb{c}(z^h(s))\mathrm{d}s, 
$$
we have
\begin{equation}\label{eq:b_averaging}
\sup_{0\le t\le T}\left|\delta^h(t)\right|\le h~T~C(\epsilon,\alpha,r)=:hTC
\end{equation}
for $0<h<\min\{H, r^{-1}\epsilon, (\epsilon/2C(r))^{1/2}\}$. 
At the same time, we have
\begin{align*}
   \delta(t):= \zeta(t)-\zeta(0)-\int_0^t \mathbf{c}(\zeta(s))\dif s=0\ (0\le t\le T). 
\end{align*}
Therefore, according to (\ref{eq:c_identity}), 
for $0<t\le T$, we have
\begin{align*}
e_t&=\sup_{0\le s\le t}|\mathbf{p}(z^h(s))-\zeta(s)|\\
&=\sup_{0\le s\le t}\left|\delta^h(t)-\delta(t)+\int_0^s\mathbf{c}(\mathbf{p}(z^h(u)))-\mathbf{c}(\zeta(u))\dif u\right|\\
   &\le htC+\sup_{0\le s\le t}\left|\int_0^s\mathbf{c}(\mathbf{p}(z^h(u)))-\mathbf{c}(\zeta(u))\dif u\right|\\
    &\le 
    htC+\left(\sup_{x\in B_r\cap \mathcal{M}_\alpha^\epsilon, v\in B_r} |\mathbf{c}'(x,v)|\right)~\int_0^te_u\dif u. 
\end{align*}
The claim follows from Gr\"onwall's lemma. 
\end{proof}

\begin{remark}\label{rem:property_of_the_limit}
The process $w(t)$ is always tangential to $\overline{\xi}(x(t))$, that is, 
\begin{equation}
\label{eq:xw}
\overline{\xi}(x(t))^{\top}w(t)=0. 
\end{equation}
This equation is true when $t=0$, and if $t>0$, we have
\begin{align*}
    \frac{\dif}{\dif t}(\overline{\xi}(x(t))^{\top}w(t))&\stackrel{~\quad~}{=}\left\{ \frac{\dif}{\dif t}\overline{\xi}(x(t))\right\}^{\top}w(t)
    +
     \overline{\xi}(x(t))^{\top}\left\{\frac{\dif}{\dif t}w(t)\right\}\\
     &\stackrel{(\ref{eq:limit_equation})}{=}
     2\partial\overline{\xi}(x(t))[w(t),w(t)]
    -2
     \overline{\xi}(x(t))^{\top}P(x(t))x(t)
     +\overline{\xi}(x(t))^{\top}b(x(t),w(t))\\
     &\stackrel{~\quad~}{=}0
\end{align*}
since $P(x(t))\overline{\xi}(x)=0$ and $\partial\overline{\xi}(x)[\overline{\xi}(x),\cdot]=0$ by (\ref{eq:p_and_q}) and (\ref{eq:limit_equation}). 
Thus (\ref{eq:xw}) holds for $t\ge 0$. 
Due to the property, 
the process $x(t)$ does not change the value of $U$ since
$$
\frac{\dif U(x(t))}{\dif t}=\xi(x(t))^{\top}\frac{\dif x(t)}{\dif t}\stackrel{(\ref{eq:limit_equation})}{=}2\xi(x(t))^{\top}w(t)\stackrel{(\ref{eq:xw})}{=}0. 
$$
Note that the process $\zeta(t)=(x(t),w(t))$ does not stay on a sphere, since it is the limit of $(x^h(t),P(x^h(t))v^h(t))$, not of $(x^h(t), v^h(t))$, where the latter remains on the sphere of radius $r=(|x^h(0)|^2+|v^h(0)|^2)^{1/2}$, but the former is inside the sphere of radius $r$ (see Figure \ref{fig::manifold}). The process $x(t)$ always stay on the manifold $\mathcal{M}_\alpha$, while $x^h(t)$ will move along some ellipses which are bounded and around the trace of $x(t)$ according to Theorem \ref{pro::3.1}.
\end{remark}

%\subsection{Geometry of the proposal transform of the Weave Metropolis kernel}

\begin{figure}[h]
\centering
\includegraphics[width=0.7\textwidth]{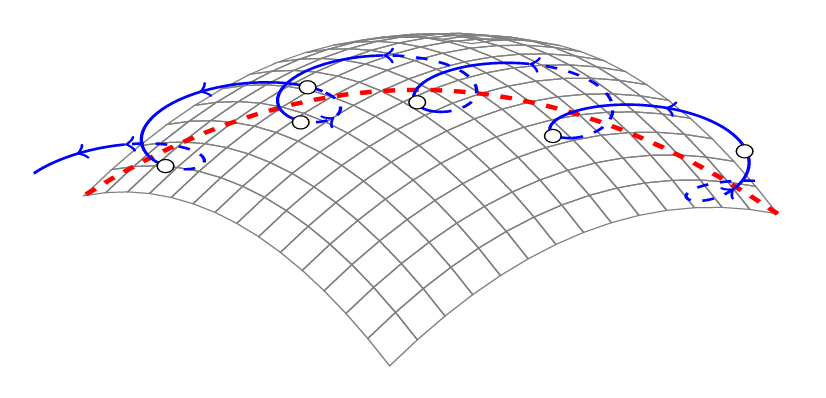}
\caption{
Illustration of the path of the process $x^h(t)$. The surface represents  the manifold $\mathcal{M}_\alpha$. The red dashed line denotes $x(t)=p_X(\zeta(t))$ and the blue line is the trace of $x^h(t)$, the  circles are the points where bounce transform occur.
}
\label{fig::manifold}
\end{figure}

\section{Simulation}\label{sec:simulation}
%To investigate the performance of proposed methods.

\subsection{Simulation settings}

In this section, we compare the performance of eight proposed kernels as described in Table \ref{table:algorithms}. The kernels are divided into two classes. The first class consists of four gradient-free, information-blind kernels. This class includes the random-walk Metropolis kernel along with the preconditioned Crank--Nicolson kernel (\textsc{pcn}), the mixed preconditioned Crank--Nicolson kernel \citep{MR3668488}  (\textsc{mpcn})
and the guided mixed preconditioned Crank--Nicolson kernel \citep{kamatani2020}  (\textsc{gmpcn}). The \textsc{pcn} kernel is the Metropolis kernel with normal reference measure, and \textsc{mpcn} is its Haar mixture counterpart, and \textsc{gmpcn} is the non-reversible version of \textsc{mpcn}. The second class consists of four gradient-based kernels. This class includes the infinite dimensional Hamiltonian Monte Carlo kernel (Example \ref{ex:idhmc}) ($\infty$-\textsc{hmc}), the Hug-and-Hop kernel \citep{ludkin2019hug} (\textsc{hh}), and the two proposed algorithms, the Weave-Metropolis kernel and the Haar-Weave-Metropolis kernel. We do not show the results of the Hamiltonian Monte Carlo, since the infinite dimensional Hamiltonian Monte Carlo always performs better in our experiments. 
All kernels except \textsc{rwm} and \textsc{hh} are Gaussian-based in the sense that the proposal distributions are invariant to the normal distribution or its mixture. 
% The infinite dimensional Hamiltonian Monte Carlo kernel is in fact a natural competitor for the Weave kernels, since it also uses circle transform. The main difference is that the infinite dimensional Hamiltonian Monte Carlo kernel uses the shift operator $v_h = v_{h/2}-h\xi(x_{h/2})/2$ 
% instead of the bounce operator $v_h = R(x_{h/2})v_{h/2}$. 

\begin{table*}[ht!]
\centering
	\caption{Markov kernels in Section \ref{sec:simulation}. The first four algorithms are gradient-free, information blind kernels. The last four algorithms are gradient based, informed kernels. }
	\label{table:algorithms}
	\begin{tabular}{ll}
		\textsc{rwm}& Random-walk Metropolis\\
		\textsc{pcn}& Preconditioned Crank--Nicolson\\
		\textsc{mpcn}& Mixed preconditioned Crank--Nicolson\\
		\textsc{gmpcn}& $\Delta$-guided mixed preconditioned Crank--Nicolson\\
		\hline
		$\infty$-\textsc{hmc}& Infinite dimensional Hamiltonian Monte Carlo\\
		\textsc{hh}& Hug-and-Hop\\
		\textsc{wm}& Weave-Metropolis\\
	    \textsc{hwm}& Haar-Weave-Metropolis
	\end{tabular}
\end{table*}

\begin{table*}[ht!]
\centering
	\caption{Abbreviations in tables in Section \ref{sec:simulation}. }
	\label{table:abbreviation}
	\begin{tabular}{ll}
		\textsc{ar}& The acceptance ratio\\
		\textsc{ess}-min& The minimum of the effective sample sizes for the coordinates\\
		\textsc{essl}& The effective sample size for the likelihood\\
		\textsc{msjd}& The mean square jump distance\\
	\end{tabular}
\end{table*}

The efficiency of the Markov kernels is compared using the effective sample size \cite[see][p 126]{liu2001monte} and the mean square jump distance \citep{MR1428751}. 
%The effective sample size is defined as the fraction between the variance of the target distribution and that of the Markov kernel. 
Let $x_1,x_2,\ldots$ be a sequence of one-dimensional stationary Markov chain with exponentially ergodic Markov kernel. The 
 effective sample size is defined by 
 $$
 N\frac{\lambda^2}{\sigma^2}
 $$
 where $\lambda^2$ is the variance of the invariant measure and
 $$
 \sigma^2=\lambda^2+2\sum_{n=1}^{\infty}\operatorname{cov}(x_1,x_{n+1})
 $$
 is the asymptotic variance of the Markov kernel corresponding to the function of interest $f(x)=x$. In our case, the Markov chain is multidimensional. Therefore, we evaluate the (approximation of the) effective sample size for each coordinate of the Markov chain and take its minimum. We also consider the effective sample size of the log likelihood with respect to the Lebesgue measure. The coordinate processes and the log likelihood processes are not Markov chains, but we expect these statistics to illustrate the efficiency of the Markov kernels. In this work, all effective sample sizes are computed using the \textit{R} package \textit{mcmcse} developed by \cite{mcmcse}. On the other hand, 
 the mean square jump distance is defined as 
 $$
 \frac{1}{N}\sum_{n=0}^{N-1}\left|x_{n+1}-x_n\right|^2.
 $$

All statistics are summarised in Table \ref{table:abbreviation}.
Note that two acceptance probabilities are given for the Hug-and-Hop kernel. This is because the Hug-and-Hop procedure consists of two types of Metropolis--Hastings procedures, resulting in two acceptance probabilities.

In the experiment, the tuning parameters of the Markov kernels are divided into three categories. 
The first category is the global parameters. More precisely 
the mean parameter $M$ and the variance parameter $\Sigma$ of the target distribution. These parameters are estimated using the adaptive Metropolis kernel, independently of the kernels in Table \ref{table:algorithms}, to allow a fair comparison. The number of iterations is fixed at $10^5$. 
These parameters are inserted into the normal distributions of the Gaussian-based kernels.  The variance parameter is also used in \textsc{rwm}, and both parameters are used in \textsc{hh}. For the Haar motion methods (\textsc{mpcn}, \textsc{gmpcn}, \textsc{hwm}), not all global parameters are fully determined at this stage because one global parameter $g$ is updated by the iteration.

The second category is the local parameters. This class includes 
the local scale parameter $s$ in $Q(x,\cdot)=\mathcal{N}_d(x, s\Sigma)$ of \textsc{rwm} and the parameter $h$ in the circle transform. These parameters are tuned using an acceptance rate criterion.   For \textsc{rwm}, the parameters are tuned so that the acceptance rate is approximately $25\%$ \citep{MR1428751}. For preconditioned Crank--Nicolson kernels, there is no known optimal criterion, but we choose acceptance rates of $30\%$ to $50\%$, which leads to better performance in our experiments. 
As suggested by \cite{Beskos_2013}, the acceptance probability of $\infty$-\textsc{hmc} is chosen around $65\%$.  For the Weave kernels, the acceptance rates are chosen around $60\%$.

For the remaining parameters, rather artificial values are set. 
Since we consider relatively high-dimensional problems, the gradient-based methods have high computational costs in our experiments. Therefore, we set the number of transforms per iteration to $L=1$ so that they are comparable to the well-tuned non-gradient-based methods. In \textsc{hh}, there are two types of transforms. Therefore, we take a single step $L=1$ for each of them. For other tuning parameters of \textsc{hh}, we follow the tuning recommendations described in \cite{ludkin2019hug}. The kernel had difficulty for our experiments since it is designed for light or super-light target distributions. %Also, it could be better since it is relatively new kernel and so the strategy for parameter tuning is still developing. 
See \cite{ludkin2019hug} for the detail. 

 We performed all experiments on a desktop computer with 6 cores Intel i7-5930K (3.50 GHz) CPU. All algorithms use the package \textit{RcppArmadillo} version 0.9.850.1.0 \cite{RcppArmadillo}. The code for all experiments is available in the online repository at the link https://github.com/Xiaolin-Song/WM.

\subsection{Logistic regression}

Posterior inference of Bayesian logistic regression with the $d$-dimensional Cauchy prior distribution $\propto (1+|x|^2)^{-(d+1)/2}\dif x$ is considered. The model is applied to two data sets, 
one is the cancer dataset and the other is the sonar dataset from  \cite{Dua:2019}. The cancer dataset contains 569 observations and 31 variables, and the sonar dataset contains 208 observations and 60 explanatory variables. As suggested in \cite{MR2655663}, all non-binary variables are scaled to have a mean of $0$ and a standard deviation of $0.5$. 
We run each algorithm for $1\times 10^6$ iterations and take the first $1\times 10^5$ iterations as burn-in. The results are shown in Tables \ref{tabl:1} and \ref{tabl:2}. 

For both tables, 
in terms of effective sample size per second, the Haar motion methods and $\infty$-\textsc{hmc} perform better than the others. This suggests that the Haar motion effectively explores this heavy-tailed distribution and $\infty$-\textsc{hmc} effectively explores the complicated likelihood surface. However, Table \ref{tabl:2} shows the advantage of the gradient information more clearly, which is related to the fact that the sonar dataset has a larger number of explanatory variables. 

The  important point to note here is that when comparing two normal reference type kernels (\textsc{pcn}, \textsc{wm}) and the three normal mixture reference type kernels (\textsc{mpcn, gmpcn, hwm}), the Weave methods perform better in each of the statistics. This indicates the advantage of the bounce transform. 

Another important observation is the comparison between $\infty$-\textsc{hmc} and \textsc{wm}. As we will see in other experiments, $\infty$-\textsc{hmc} always outperforms \textsc{wm} in terms of effective sample size. This is natural since \textsc{wm} does not change the likelihood as much as described in Remark \ref{rem:property_of_the_limit}. On the other hand, the mean square jump per second illustrates the limitation of $\infty$-\textsc{hmc}. This is probably due to the fact that for a large $|x|$, the kernel behaves like the random-walk Metropolis kernel, since the contribution of the gradient is relatively small when $x$ is far from the origin. On the other hand, the \textsc{wm} effectively uses the directional information even when $|x|$ is large. 

While \textsc{wm} has difficulty updating the likelihood, \textsc{hwm} shows the best performance for all statistics, even better than $\infty$-\textsc{hmc}, thanks to Haar motion.

\renewcommand{\tabcolsep}{3pt} % default: 6pt
\begin{table}
\caption{The performance for the logistic regression result with Cancer data}
\begin{tabularx}{\textwidth}{@{} C{0.2} C{0.2} C{0.2} C{0.2}C{0.2} C{0.2} C{0.2} C{0.2} C{0.2} @{}} 
%\toprule
Methods & \textsc{essl}  &\textsc{ess}-min & \textsc{msjd} & \textsc{essl}/s &  \textsc{ess}-min/s & \textsc{msjd}/s & Time & \textsc{ar} \\
\midrule
\textsc{rwm}& 1906.38 & 5001.71 & 108.56 & 33.06 & 86.74 & 1.88 & 57.66 & 0.20 \\ 
  \textsc{pcn}& 1970.02 & 3554.88 & 4542.60 & 32.85 & 59.28 & 75.75 & 59.97 & 0.31 \\ 
  \textsc{mpcn}& 18412.10 & 33675.57 & 5505.95 & 311.18 & 569.15 & 93.05 & 59.17 & 0.35 \\ 
  \textsc{gmpcn}& 26227.38 & 46693.52 & 5507.35 & 433.63 & 772.00 & 91.06 & 60.48 & 0.35 \\ 
  $\infty$-\textsc{hmc}& 46175.15 & 86752.44 & 5062.57 & 437.12 & 821.25 & 47.93 & 105.63 & 0.63 \\ 
  \textsc{hh}& 22263.42 & 49639.42 & 2115.66 & 79.62 & 177.53 & 7.57 & 279.62 & 0.83/0.25  \\ 
  \textsc{wm}& 8987.86 & 36601.77 & 14302.22 & 89.23 & 363.36 & 141.99 & 100.73 & 0.62 \\ 
  \textsc{hwm}& 48826.31 & 140611.82 & 20164.62 & 471.93 & 1359.08 & 194.90 & 103.46 & 0.64 \\ 
%  \textsc{hh}(1,1) & 19003.89 & 50476.16 & 67.74 & 179.92 & 280.55 & -0.86 & 2205.28 \\ 
%   hughop(5,1) & 23703.05 & 72076.55 & 51.28 & 155.93 & 462.25 & 0.89 & 79966.45 \\ 
%   hughop(5,2) & 45339.73 & 227823.77 & 71.19 & 357.72 & 636.87 & 0.92 & 79136.22 \\ 
%\bottomrule
\end{tabularx} \label{tabl:1}
\end{table}

\renewcommand{\tabcolsep}{3pt} % default: 6pt
\begin{table}[H]
\caption{The performance for the logistic regression result with Sonar data}
\begin{tabularx}{\textwidth}{@{} C{0.2} C{0.2} C{0.2} C{0.2}C{0.2} C{0.2} C{0.2} C{0.2} C{0.2} @{}} 
%\toprule
Methods & \textsc{essl}  &\textsc{ess}-min & \textsc{msjd} & \textsc{essl}/s &  \textsc{ess}-min/s & \textsc{msjd}/s & Time & \textsc{ar} \\

\midrule
\textsc{rwm}& 790.29 & 1782.22 & 0.84 & 26.88 & 60.61 & 0.03 & 29.40 & 0.23 \\ 
  \textsc{pcn}& 1366.79 & 2552.64 & 23.79 & 42.85 & 80.03 & 0.75 & 31.90 & 0.28 \\ 
  \textsc{mpcn}& 6364.51 & 7294.80 & 49.71 & 209.53 & 240.16 & 1.64 & 30.38 & 0.26 \\ 
  \textsc{gmpcn}& 11819.67 & 18009.25 & 50.17 & 361.32 & 550.53 & 1.53 & 32.71 & 0.27 \\ 
  $\infty$-\textsc{hmc}& 38026.66 & 60822.82 & 160.52 & 623.98 & 998.05 & 2.63 & 60.94 & 0.66 \\ 
  \textsc{hh}& 17923.04 & 35698.99 & 52.98 & 64.98 & 129.43 & 0.19 & 275.81 & 0.83/0.34 \\ 
  \textsc{wm}& 6076.32 & 16777.28 & 318.16 & 102.96 & 284.29 & 5.39 & 59.01 & 0.68 \\ 
  \textsc{hwm}& 41157.92 & 89597.69 & 442.81 & 690.57 & 1503.32 & 7.43 & 59.60 & 0.65 \\ 
%\textsc{hh}(1,1) & 16325.01 & 35759.34 & 59.64 & 130.64 & 273.73 & 0.91 & 48.24 \\ 
%   hughop(5,1) & 24771.85 & 97044.79 & 62.72 & 245.70 & 394.97 & 0.93 & 2221.18 \\ 
%   hughop(5,2) & 48781.20 & 195313.29 & 80.84 & 323.67 & 603.43 & 0.96 & 2206.64 \\ 

%\bottomrule
\end{tabularx}\label{tabl:2}
\end{table}

\subsection{The stochastic volatility model}
Next, we consider the stochastic volatility model. In this example, we consider the sampling of the latent variables in the following stochastic volatility model. For a positive integer $T$, for $t=1,\ldots, T$, let 
\begin{equation}
\nonumber
\begin{split}
x_t&=\phi~x_{t-1}+\epsilon_t, \quad\epsilon_t\sim\mathcal{N}(0,\sigma^2)\\
y_t&=\exp\left(\frac{x_t}{2}\right)~w_t,\quad w_t\sim\mathcal{N}(0,1)\\
x_0&=\frac{\phi}{1-\phi^2}~w_0,\quad w_0\sim\mathcal{N}(0,\sigma^2)
\end{split}
\end{equation}
where $\epsilon_1,\epsilon_2,\ldots$ and $w_0, w_1$ are independent. 
The variables $y_1, y_2,\ldots$ are observed, while $x_0, x_1,\ldots$ are not observed. The parameters of interest are the mean inversion parameter $\phi\sim \mathcal{B}e(2,5)$ and the standard deviation $\sigma\sim\mathcal{G}(5,0.2)$. In our simulation, the number of observations is $T=100$ with $\phi=0.5$ and $\sigma=10$.

 The result is shown in Table \ref{table:3}. It seems that gradient-based methods are better for this example, which is probably due to the fact that the likelihhod in this case is relatively cheap to evaluate.  We would like to note that when comparing  \textsc{mpcn} and  \textsc{hwm}, the Haar motion itself is not sufficient to explore the likelihood surface. On the other hand, the bounce transform alone is not sufficient when comparing $\infty$-\textsc{hmc} and some bounce kernels (\textsc{hh}, \textsc{wm}). The combination is very effective, as shown by the performance of the \textsc{hwm}.

\renewcommand{\tabcolsep}{2pt} % default: 6pt
\begin{table}
\caption{The performance for the stochastic volatility model}
\begin{tabularx}{\textwidth}{@{} C{0.2} C{0.2} C{0.2} C{0.2}C{0.2} C{0.2} C{0.2} C{0.2} C{0.2} @{}} 
%\toprule
Methods & \textsc{essl}  &\textsc{ess}-min & \textsc{msjd} & \textsc{essl}/s &  \textsc{ess}-min/s & \textsc{msjd}/s & Time & \textsc{ar} \\
\midrule
 \textsc{rwm} & 865.27 & 964.13 & 33.46 & 46.42 & 51.73 &1.80 & 18.64 & 0.20 \\ 
 \textsc{pcn} & 1128.64 & 890.82 & 35.00 & 60.38 & 47.66 & 1.87 & 18.69 & 0.37 \\ 
    \textsc{mpcn} & 3159.22 & 1783.71 & 76.93 & 164.08 & 92.64 & 4.00 & 19.25 & 0.34 \\ 
    \textsc{gmpcn} & 3747.38 & 1975.54 & 100.17 & 160.97 & 84.86 & 4.30 & 23.28 & 0.28 \\ 
    $\infty$-\textsc{hmc}& 17861.48 & 14655.91 & 2105.46 & 465.30 & 381.80 & 54.84 & 38.39 & 0.63 \\ 
    \textsc{hh} & 11104.22 & 18234.27 & 3415.83 & 18.21 & 29.90 & 5.60 & 609.77 & 0.52/0.23 \\ 
    \textsc{wm} & 1772.11 & 8027.00 & 3206.11 & 46.07 & 208.66 & 83.34 & 38.47 & 0.62 \\ 
    \textsc{hwm} & 22531.37 & 18314.68 & 3871.59 & 581.06 & 472.31 & 99.83 & 38.78 & 0.59 \\ 

%  \textsc{hh}(1,1) & 7397.20 & 6473.54 & 12.04 & 10.54 & 614.29 & 0.59 & 3076.84 \\
%   hughop(5,1) & 15525.81 & 26181.01 & 22.10 & 37.26 & 702.66 & 0.76 & 224766.48 \\ 
%   hughop(5,2) & 31269.40 & 1095.55 & 24.84 & 0.87 & 1258.58 & 0.83 & 238463.24 \\ 

%\bottomrule
\end{tabularx}\label{table:3}
\end{table}

\subsection{Discrete observation of the stochastic diffusion process}
\label{5.1.1}
Finally, we consider statistical inference for the stochastic process driven by the Wiener process. 
The motivation for this model is to study a scenario with a very complicated high-dimensional target distribution. In this case, the gradient-based methods are likely to be less effective. 

Let $k, d$ and $N$ be positive integers, $T$ be a positive number, and let $\alpha\in\mathbb{R}^k$. 
Suppose $(X_t)_{t\in [0,T]}$ is a solution process of a stochastic differential equation
$$
\dif X_t=a(X_t,\alpha)~\dif t+b(X_t)~\dif W_t; X_0=x_0
$$
where $(W_t)_{t\in [0,T]}$ is the $d$-dimensional standard Wiener process and $a:\mathbb{R}^d\times\mathbb{R}^k\rightarrow \mathbb{R}^d$ and 
$b:\mathbb{R}^d\rightarrow \mathbb{R}^{d\times d}$ are the drift and diffusion coefficients, respectively. We have only one discrete observation from the path $(X_t)_{t\in [0,T]}$. 
For simplicity, we assume equally time spaced observations $X_0, X_h, X_{2h},\ldots, X_{Nh}$, where $h=T/N$. 

In general, the likelihood is not available for the solution process of the stochastic differential equation. The local Gaussian approximation has been extensively studied in the past, such as \cite{Prakasa_Rao_1983,MR999016,Florens_zmirou_1989,YOSHIDA1992220}. 
We follow this approach and consider a Bayesian inference using an approximated likelihood. 

For our experiment, we set $d=50$, $N = 10^2$, and $T=5$. 
For simplicity, we consider the case of a constant diffusion coefficient, i.e. $b\equiv 1$. The drift coefficient is
$$
a(x,\alpha)=-\frac{1}{2}\nabla V(x-\alpha);
$$
where $V(x)= 27.5\log(1+x^{\top}\Sigma^{-1}x)$. 
The symmetric positive definite matrix $\Sigma$ is generated from the Wishart distribution with $50$ degrees of freedom and the identity matrix as the scale matrix. The prior distribution of the parameter $\alpha$ is the multivariate Student $t$-distribution $\mathcal{T}_{50}(3, 0,10I_{50})$.

As mentioned earlier, gradient evaluation is expensive. Therefore, the non-gradient based methods show good performance, with the exception of the random-walk Metropolis kernel.  In particular, \textsc{mpcn} shows the best performance for the effective sample size per second. The \textsc{hwm} is the best at the mean square jump distance thanks to the Haar motion. The \textsc{wm} kernel and $\infty$-\textsc{hmc} have similar performance for effective sample size of log likelihood, but the \textsc{wm} performs better for other statistics per second. 

% The Haar-based kernels (\textsc{mpcn}, \textsc{gmpcn}, \textsc{hwm}) show relatively good performance. On the other hand, $\infty$-\textsc{hmc} does not perform well because the gradient evaluation is not informative but expensive.  Even in this complicated scenario, the Haar-Weave-Metropolis kernel performs well thanks to the combination of bounce transform and Haar movement. 

\renewcommand{\tabcolsep}{3pt} % default: 6pt
\begin{table}
\caption{The performance of the stocahstic differential equation}
\begin{tabularx}{\textwidth}{@{} C{0.2} C{0.2} C{0.2} C{0.2}C{0.2} C{0.2} C{0.2} C{0.2} C{0.2} @{}} 
%\toprule
Methods & \textsc{essl}  &\textsc{ess}-min & \textsc{msjd} & \textsc{essl}/s &  \textsc{ess}-min/s & \textsc{msjd}/s & Time & \textsc{ar} \\
\hline
\textsc{rwm}& 2034.12 & 2799.71 & 310.46 & 6.49 & 8.94 & 0.99 & 313.18 & 0.31 \\ 
  \textsc{pcn}& 12897.57 & 12785.26 & 850.10 & 40.99 & 40.63 & 2.70 & 314.68 & 0.43 \\ 
  \textsc{mpcn}& 16548.68 & 21805.46 & 11392.07 & 52.12 & 68.67 & 35.88 & 317.53 & 0.30 \\ 
  \textsc{gmpcn}& 23184.34 & 22797.25 & 11445.11 & 71.48 & 70.28 & 35.28 & 324.37 & 0.30 \\ 
  $\infty$-\textsc{hmc}& 15252.22 & 13420.56 & 3383.18 & 16.25 & 14.30 & 3.60 & 938.74 & 0.69 \\ 
  \textsc{hh}& 10061.25&29574.45 &12093.45 &3.93&11.56& 4.73 & 2557.61& 0.52/0.20 \\ 
  \textsc{wm}& 14260.80 & 45389.57 & 26949.13 & 15.54 & 49.47 & 29.37 & 917.47 & 0.64 \\ 
  \textsc{hwm}& 44554.04 & 61870.96 & 41298.56 & 48.31 & 67.09 & 44.78 & 922.19 & 0.63 \\ 
%  \textsc{hh}(1,1) & 9106.89 & 30762.17 & 3.58 & 12.11 & 2540 & 0.62 & 12085.55 \\ 
%   hughop(5,1) & 8127.23 & 109528.14 & 1.65 & 22.29 & 4910 & 0.71 & 847521.34 \\ 
%   hughop(5,2) & 18181.96 & 315925.97 & 2.80 & 48.75 & 6480 & 0.77 & 836464.94 \\ 

%\bottomrule
\end{tabularx}
\end{table}

\section{Discussion}

% In the above experiments, we let the step number $L=1$ for the purpose of comparison. However, to take the advantage of gradient base method, small step size and multiple step number is always required. 
We introduced a new algorithm, the Weave-Metropolis kernel, which is based on the Weave transform and combines circular and bounce transforms. The Haar motion lifts the kernel more efficiently, especially for heavy-tailed target distributions. The Weave transform has similarities with Hamiltonian flow. The former does not change the potential energy as much (see Remark \ref{rem:property_of_the_limit}) and the latter does not change the Hamiltonian. However, the former can be combined with the Haar motion, while the latter cannot. 
 The results of the numerical experiment show that the new method is efficient and robust due to the local property of the transform and the global property of the Haar motion. 

We would like to remark about parameter tuning. The effects of the number of iterations per step $L$ for the Weave kernels are different from the Hamiltonian Monte Carlo kernels. For the Hamiltonian Monte Carlo kernel, a large $L$ usually leads to a large mean square jump distance and a small acceptance probability. However, due to the circular transform, the effect is non-monotonic for the Weave kernels. To obtain the best performance, $L$ and also the step size $h$ should be carefully tuned. 
Second, it is advisable to randomise the step size of the circular transform, since the transform may be periodic. This does not mean that the Markov kernel is reducible, since $v$ is refreshed at each iteration. However, this periodicity can slow down convergence. Similar phenomena can occur in the Hamiltonian Monte Carlo method \citep[see p. 127 in][]{brooks2011}.

Finally, we would like to discuss possible extensions of the results.  One approach is to combine various transforms with the Haar measure. For example, in \cite{kamatani2020} they studied the Beta-Gamma Haar mixture on $\mathbb{R}_+^d$, which can be combined with gradient information. 
Also, due to the similarity with the Hamiltonian Monte Carlo kernel, we can use the same or similar techniques that can improve the performance of the Hamiltonian Monte Carlo kernel, e.g., parallel computation in \cite{yang2018}, delayed acceptance in \cite{park2020markov}. This technique can further improve performance.

\subsection*{Acknowledgement}

KK and XS were supported by JST CREST Grant number JPMJCR14D7. KK was supported by JSPS KAKENHI Grant number 20H04149 and 21K18589. XS was supported by the Ichikawa International Scholarship Foundation.

\appendix

\section{Algorithms}
\begin{algorithm}[H]
	\caption{Weave-Metropolis} 
		 For target distribution
		 $\Pi(\dif x)=\exp(-U(x))\mu(\dif x)$ where $\mu=\mathcal{N}_d(M,\Sigma)$. \\
		 \hspace*{\algorithmicindent}
		 \textbf{Input}: current state $x$, pre-conditional matrix $\Sigma$, step number $L$, step size $h$
       
	\begin{algorithmic}[1]
        \State Generate $v \sim \mu$ 
        \State
        $(x_0,v_0)=(x,v)$
     %   \State Generate $L \sim Unif\{1,M\}$
	\For {$l=1,2,\ldots L$}
	    \State $(x_{l-\frac{1}{2}},v_{l-\frac{1}{2}})=\Phi_{\Circle}(x_{l-1},v_{l-1}\mid M)$   \Comment{circle transform}
		\State $(x_{l-\frac{1}{2}},{v}_{l-\frac{1}{2}}^*)=\Phi_{\Bounce}(x_{l-\frac{1}{2}},v_{l-\frac{1}{2}}\mid M,\Sigma)$ \Comment{bounce transform}
		\State $(x_l,v_l)=\Phi_{\Circle}(x_{l-\frac{1}{2}},v_{l-\frac{1}{2}}^*\mid M)$   \Comment{circle transform}
		\EndFor 
		%\State $v_L=-v_L$   
		\State Generate $w \sim \mathcal{U}[0,1]$
		\State Compute the acceptance rate 
	$\alpha(x,x_L)=\min\{1,\exp(-U(x_L)+U(x_0))\}
				$
		\If{$w\ge \alpha(x, x_L)$}
		
            $x_{new}= x_L$
        \Else
            \State $x_{new}= x$
        \EndIf
		%\State $\theta_{old}\leftarrow\theta$
	\end{algorithmic} 
	\hspace*{\algorithmicindent} \textbf{Output}: new state $x_{new}$ 
	\label{al:wm}
\end{algorithm}

\begin{algorithm}[H]
	\caption{Haar-Weave-Metropolis}
	The target distribution is
		 $\Pi(\dif x)=\exp(-U(x))\mu_*(\dif x)$, where $\mu_*(\dif x)=(\Delta x)^{-d/2}\dif x$ with $\Delta x=(x-M)^{\top}\Sigma^{-1}(x-M)$. \\
		 \hspace*{\algorithmicindent} \textbf{Input}: current state $x$, pre-conditional matrix $\Sigma$, step number $L$, step size $h$
       
	\begin{algorithmic}[1]
	    \State Generate $g \sim \mathcal{G}(d/2,\Delta x/2)$, $v \sim \mathcal{N}_d(M,g^{-1}~\Sigma)$  \Comment{Haar motion}
        \State
         $(x_0,v_0)=(x,v)$
       % \State Generate $L \sim Unif\{1,M\}$
	\For {iteration $l=1,2,\ldots L$}
	    \State $(x_{l-\frac{1}{2}},v_{l-\frac{1}{2}})=\Phi_{\mathrm{circle}}(x_{l-1},v_{l-1}\mid M)$   \Comment{circle transform}
		\State $(x_{l-\frac{1}{2}},{v}_{l-\frac{1}{2}}^*)=\Phi_{\mathrm{bounce}}(x_{l-\frac{1}{2}},v_{l-\frac{1}{2}}\mid M,\Sigma)$ \Comment{bounce transform}
		\State $(x_l,v_l)=\Phi_{\mathrm{circle}}(x_{l-\frac{1}{2}},v_{l-\frac{1}{2}}^*\mid M)$   \Comment{circle transform}
		\EndFor
		%\State $v_L=-v_L$  
		\State Generate $w \sim \mathcal{U}[0,1]$
		\State Compute the acceptance rate 
	$\alpha(x,x_L)=\min\{1,\exp(-U(x_L)+U(x_0))\}
				$
		\If{$w\ge \alpha(x,x_L)$}
		
            $x_{new}= x_L$
        \Else
            \State $x_{new}= x_0$
        \EndIf
		%\State $\theta_{old}\leftarrow\theta$
	\end{algorithmic} 
	\hspace*{\algorithmicindent} \textbf{Output}: new state $x_{new}$ 
	\label{al:hwm}
\end{algorithm}

\bibliographystyle{apalike}

\end{document}